\documentclass[english]{elsarticle}
\usepackage[T1]{fontenc}
\usepackage[latin9]{inputenc}
\usepackage{geometry}
\usepackage{units}
\usepackage{mathrsfs}
\usepackage{url}
\usepackage{amsmath}
\usepackage{amssymb}
\usepackage{graphicx}

\makeatletter

\providecommand{\tabularnewline}{\\}

\newenvironment{lyxlist}[1]
{\begin{list}{}
{\settowidth{\labelwidth}{#1}
 \setlength{\leftmargin}{\labelwidth}
 \addtolength{\leftmargin}{\labelsep}
 }}
{\end{list}}
 \newenvironment{proofQED}{\begin{proof}}{\qed\end{proof}}

\RequirePackage{fix-cm}

\usepackage{todonotes}

\usepackage{transparent}

\newtheorem{theorem}{Theorem}
\newtheorem{definition}{Definition}
\newtheorem{lemma}{Lemma}
\newdefinition{remark}{Remark}
\newdefinition{corollary}{Corollary}
\newproof{proof}{Proof}
\newproof{thesis}{Thesis}
\newtheorem{proposition}{Proposition}

\newcommand*{\QED}{\null\hfill$\square$}

\usepackage{tikz}

\@ifundefined{showcaptionsetup}{}{%
 \PassOptionsToPackage{caption=false}{subfig}}
\usepackage{subfig}
\makeatother

\usepackage{babel}
\begin{document}

\title{Inconsistency in the ordinal pairwise comparisons method \\
with and without ties}

\author{Konrad Ku\l akowski}

\ead{konrad.kulakowski@agh.edu.pl}

\address{AGH University of Science and Technology, Krak�w, Poland}
\begin{abstract}
Comparing alternatives in pairs is a well-known method of ranking
creation. Experts are asked to perform a series of binary comparisons
and then, using mathematical methods, the final ranking is prepared.
As experts conduct the individual assessments, they may not always
be consistent. The level of inconsistency among individual assessments
is widely accepted as a measure of the ranking quality. The higher
the ranking quality, the greater its credibility. 

One way to determine the level of inconsistency among the paired comparisons
is to calculate the value of the inconsistency index. One of the earliest
and most widespread inconsistency indices is the consistency coefficient
defined by Kendall and Babington Smith. In their work, the authors
consider binary pairwise comparisons, i.e., those where the result
of an individual comparison can only be: better or worse. The presented
work extends the Kendall and Babington Smith index to sets of paired
comparisons with ties. Hence, this extension allows the decision makers
to determine the inconsistency for sets of paired comparisons, where
the result may also be \textquotedbl{}equal.\textquotedbl{} The article
contains a definition and analysis of the most inconsistent set of
pairwise comparisons with and without ties. It is also shown that
the most inconsistent set of pairwise comparisons with ties represents
a special case of the more general set cover problem.
\end{abstract}
\maketitle

\section{Introduction\label{sec:Introduction}}

The use of pairwise comparisons (PC) to form judgments has a long
history. Probably the first who formally defined and used pairwise
comparisons for decision making was \emph{Ramon Llull} (the XIII century)
\cite{Colomer2011rlfa}. He proposed a voting system based on binary
comparisons. The subject of comparisons (alternatives) were people
- candidates for office. Voters evaluated the candidates in pairs,
deciding which one was better. In the XVIII century, \emph{Llull's}
method was rediscovered by \emph{Condorcet} \cite{Condorcet1785eota},
then once again reinvented in the middle of the XX century by \emph{Copeland}
\cite{Colomer2011rlfa,Copeland1951arsw}. At the beginning of the
XX century, \emph{Thurstone} used the pairwise comparisons method
(PC method) quantitatively \cite{Thurstone1927tmop}. In this approach,
the result returned does not only contain information about who or
what is better, but also indicates how strong the preferences are.
Later, both approaches, ordinal (qualitative), as proposed by \emph{Llull,}
and cardinal (quantitative), as used by \emph{Thurstone}, were developed
in parallel. Comparing alternatives in pairs plays an important role
in research into decision making systems \cite{Greco2005mcda,Greco2010dbrs,Kulakowski2015hreg},
ranking theory \cite{Saaty1977asmf,Janicki2012oapc}, social choice
theory \cite{Suzumura2010hosc}, voting systems \cite{Tideman1987ioca,Faliszewski2009lacv,Vargas2013vwio}
and others. 

In general, the PC method is a ranking technique that allows the assessment
of the importance (relevance, usefulness, competence level etc.) of
a number of alternatives. As it is much easier for people to assess
two alternatives at a time than handling all of them at once, the
PC method assumes that, first, all the alternatives are compared in
pairs, then, by using an appropriate algorithm, the overall ranking
is synthesized. The choice of the algorithm is not easy and is still
the subject of research and vigorous debate \cite{Saaty1998rbev,Wang2007peit,Kulakowski2015otpo}.
Of course, it also depends on the nature of the comparisons. The cardinal
methods use different algorithms \cite{Ishizaka2006htdp,Fedrizzi2010otpv}
than the ordinal ones \cite{Janicki2012oapc,Colomer2011rlfa,Janicki1996awoa,Tideman1987ioca}.
Despite the many differences between ordinal and cardinal pairwise
comparisons, both approaches have much in common. For example, both
approaches use the idea of inconsistency among individual comparisons.
The notion of inconsistency introduced by the pairwise comparisons
method is based on the natural expectation that every two comparisons
of any three different alternatives should determine the third possible
comparison among those alternatives. 

To better understand the phenomenon of inconsistency, let us assume
that we have to compare three alternatives $c_{1},c_{2}$ and $c_{3}$
with respect to the same criterion. If after the comparison of $c_{1}$
and $c_{2}$ it is clear to us that $c_{2}$ is more important than
$c_{2}$, and similarly, after comparing $c_{2}$ and $c_{3}$ it
is evident that $c_{3}$ is more important than $c_{1}$ then we may
expect that $c_{3}$ will also turn out to be more important than
$c_{1}.$ The situation in which $c_{1}$ is better than $c_{3}$
would raise our surprise and concern. That is because it seems natural
to assume that the preferential relationship should be transitive.
If it is not, we have to deal with inconsistency. As pairwise comparisons
are performed by experts, who, like all human beings, sometimes make
mistakes, the phenomenon of inconsistency is something natural. The
ranking synthesis algorithm must take it into account. On the other
hand, if a large number of such ``mistakes'' can be found in the
set of paired comparisons, one can have reasonable doubts as to the
credibility of the ranking obtained from such lower quality data. 

Both ordinal and cardinal PC methods developed their own solutions
for determining the degree of inconsistency. Research into the cardinal
PC method resulted in a number of works on inconsistency indices.
Probably the most popular inconsistency index was defined by \emph{Saaty}
in his seminal work on \emph{the Analytic Hierarchy Process} \emph{(AHP)}
\cite{Saaty1977asmf}. His work prompted others to continue the research
\cite{Koczkodaj1993ando,Pelaez2003anwo,Aguaron2003tgci,Stein2007thci,Bozoki2011albi,Brunelli2013iifp}.
The ordinal PC methods also have their own ways of assessing the level
of inconsistency. In their seminal work \cite{Kendall1940otmo} \emph{Kendall}
and \emph{Babington Smith} introduced the \emph{inconsistency index}
(called by the authors the \emph{consistency coefficient}). Their
index allows the inconsistency degree of a set composed of binary
pairwise comparisons to be determined. The results obtained by the
authors were the inspiration for many other researchers in different
fields of science \cite{Kadane1966seci,Maas1995osii,Parizet2002pclt,Bozoki2013aopc,Brunelli2016otce,Siraj2015coij}.

Although the ordinal pairwise comparisons method is a really powerful
and handy tool facilitating the right decision, in practice we very
often face the problem that the two options seem to be equally important.
In such a situation, we can try to get around the problem by a brute
force method of breaking ties. For example, we can do this by \emph{``instructing
the judge to toss a mental coin when he cannot otherwise reach a decision;
or, allowing him the comfort of reserving judgment, we can let a physical
coin decide for him''} \cite[p. 94 - 95]{David1969tmop}. It is clear,
however, that instead of relying on more or less arbitrary methods
of breaking ties, it is better to accept their existence and incorporate
them into the model. Indeed, ties have been inextricably linked with
the ranking theory for a long time \cite{Colomer2011rlfa,Kendall1945ttot,David1969tmop}.
The ordinal pairwise comparisons method with ties has its own techniques
of synthesizing ranking \cite{Glenn1960tipc,Davidson1970oetb,Tideman1987ioca}.
In this perspective, research into the inconsistency of ordinal pairwise
comparisons with ties is quite poor. In particular, the \emph{consistency
coefficient} as defined by \cite{Kendall1940otmo} is not suitable
for determining the inconsistency of PC with ties. The problem was
recognized by \emph{Jensen} and \emph{Hicks} \cite{Jensen1993odaa},
and later by \emph{Iida} \cite{Iida2009octa}. These authors also
made attempts to patch up this gap in the ranking theory, however,
the fundamental question as to what extent the set of PC with ties
can be inconsistent still remains unanswered. 

The purpose of the present article is to answer this question, and
thus to define the inconsistency index for the ordinal PC with ties
in the same manner as \emph{Kendall} and \emph{Babington Smith} did
\cite{Kendall1940otmo} for binary PC. The definition of the inconsistency
index is accompanied by a thorough study of the most inconsistent
sets of pairwise comparisons with and without ties. 

The article is composed of eight sections including the introduction
and four appendices. The PC with ties is formally introduced in the
next section (Sec. \ref{sec:Model-of-inconsistency}). For the purpose
of modeling PC with ties, a generalized tournament graph has also
been defined there. The most inconsistent set of binary PC is studied
in (Sec. \ref{sec:The-most-inconsistent}). It is also proven that
the number of inconsistent triads in such a graph is determined by
Kendall Babington Smith's\emph{ consistency coefficient. }The next
section (Sec. \ref{sec:Properties-of-the}) describes how the most
inconsistent set of PC with ties may look. Thus, it contains several
theorems describing the quantitative relationship between the elements
of the generalized tournament graph. Finally, in (Sec. \ref{sec:The-most-inconsistent-1})
the most inconsistent set of PC with ties is proposed. The generalized
inconsistency index for ordinal PC is also defined (Sec. \ref{sec:Inconsistency-indexes-in}).
The penultimate section (Sec. \ref{sec:Discussion-and-remarks}) contains
a discussion of the subject. In particular, the relationship between
the maximally inconsistent set of PC and the \emph{NP-complete} set
cover problem \cite{Karp1972racp} is shown. A brief summary is provided
in (Sec. \ref{sec:Summary}). 

\section{\label{sec:Model-of-inconsistency}Model of inconsistency}

Let us suppose we have a number of possible choices (alternatives,
concepts) $c_{1},\ldots,c_{n}$ where we are able to decide only whether
one is better (more preferred) than the other or whether both alternatives
are equally preferred. In the first case, we will write that $c_{i}\prec c_{j}$
to denote that $c_{j}$ is more preferred than $c_{i}$, whilst in
the second case, to express that two alternatives $c_{i}$ and $c_{j}$
are equally preferred we write $c_{i}\sim c_{j}$. The preference
relationship is total. Hence, for every two $c_{i}$ and $c_{j}$
it holds that either $c_{i}\prec c_{j}$, $c_{j}\prec c_{i}$ or $c_{i}\sim c_{j}$.
The relationship is reflexive and asymmetric. In particular, we will
assume that if $c_{i}\prec c_{j}$ then not $c_{j}\prec c_{i}$, and
$c_{i}\sim c_{i}$ for every $i,j=1,\ldots,n$. It is convenient to
represent the relationship of preferences in the form of an $n\times n$
matrix. 
\begin{definition}
\label{def-the-ordinal-pc-matrix-with-ties}The $n\times n$ matrix
$M=[m_{ij}]$ where $m_{ij}\in\{-1,0,1\}$ is said to be the \emph{ordinal
PC} \emph{matrix} for $n$ alternatives $c_{1},\ldots,c_{n}$ if a
single comparison $m_{ij}$ takes the value $1$ when $c_{i}$ wins
with $c_{j}$ (i.e. $c_{i}\succ c_{j})$, $-1$ if, reversely, $c_{j}$
is better than $c_{i}$ (i.e. $c_{j}\succ c_{i})$ and $0$ in the
case of a tie between $c_{i}$ and $c_{j}$ ($c_{i}\sim c_{j}$).
The the diagonal values are $0$.
\end{definition}

The \emph{PC matrix} is skew-symmetric except the diagonal, so that
for every $i,j=1,\ldots,n$ it holds that $m_{ij}+m_{ji}=0$. An example
of the ordinal PC matrix for five alternatives is given below (\ref{eq:A_example}).
\begin{equation}
M=\left(\begin{array}{ccccc}
0 & 1 & 0 & 1 & 0\\
-1 & 0 & 1 & 1 & 1\\
0 & -1 & 0 & 1 & -1\\
-1 & -1 & -1 & 0 & 1\\
0 & -1 & 1 & -1 & 0
\end{array}\right)\label{eq:A_example}
\end{equation}

The \emph{PC matrix} can be easily represented in the form of a graph. 
\begin{definition}
\label{def:t-graph-definition}A tournament graph \emph{(t-graph)}
with $n$ vertices is a pair $T=(V,E_{d})$ where $V=c_{1},\ldots,c_{n}$
is a set of vertices and $E_{d}\subset V^{2}$ is a set of ordered
pairs called directed edges, so that for every two distinct vertices
$c_{i}$ and $c_{j}$ either $(c_{i},c_{j})\in E_{d}$  or $(c_{j},c_{i})\in E_{d}$
.
\end{definition}

Let us expand the definition of a tournament graph so that it can
also model the collection of pairwise comparisons with ties. 
\begin{definition}
\label{def:gt-graph-definition}The generalized tournament graph \emph{(gt-graph)}
with $n$ vertices is a triple $G=(V,E_{u},E_{d})$ where $V=c_{1},\ldots,c_{n}$
is a set of vertices, $E_{u}\subset2^{V}$ is a set of unordered pairs
called undirected edges, and $E_{d}\subset V^{2}$ is a set of ordered
pairs called directed edges, so that for every two distinct vertices
$c_{i}$ and $c_{j}$ either $(c_{i},c_{j})\in E_{d}$  or $(c_{j},c_{i})\in E_{d}$
or $\{c_{i},c_{j}\}\in E_{u}$ . 
\end{definition}

Wherever it increases the readability of the text the directed and
undirected edges $(c_{i},c_{j})$, $(c_{j},c_{i})$, $\{c_{i},c_{j}\}$
between $c_{i},c_{j}\in V$ are denoted as $c_{i}\rightarrow c_{j},c_{j}\rightarrow c_{i}$
and $c_{i}-c_{j}$ correspondingly. 

It is easy to see that every tournament graph can easily be extended
to a generalized tournament graph where $E_{u}=\emptyset$. Therefore,
it will be assumed that every \emph{t-graph} is also a \emph{gt-graph,}
but not reversely. 
\begin{definition}
A family of \emph{t-graphs} with $n$ vertices will be denoted as
$\mathscr{T}_{n}^{t}$, where $\mathscr{T}_{n}^{t}=\{(V,E_{d})\,\textit{ is a t-graph, where}\,\left|V\right|=n\}$,
and similarly, a family of \emph{gt-graphs} with $n$ vertices will
be denoted as $\mathscr{T}_{n}^{g}$, where $\mathscr{T}_{n}^{g}=\{(V,E_{u},E_{d})\,\textit{ is a gt-graph, where}\,\left|V\right|=n\}$
\end{definition}

It is obvious that for every $n>0$ it holds that $\mathscr{T}_{n}^{t}\subsetneq\mathscr{T}_{n}^{g}$. 
\begin{definition}
A family of \emph{gt-graphs} with $n$ vertices and $m$ directed
edges will be denoted as $\mathscr{T}_{n,m}^{g}=\{(V,E_{u},E_{d})\,\textit{ is a gt-graph, where}\,\left|V\right|=n\,\textit{and}\,\left|E_{d}\right|=m\}$
\end{definition}

\begin{definition}
A \emph{gt-graph} $G_{M}\in\mathscr{T}_{n}^{g}$ is said to correspond
to the\emph{ $n\times n$ ordinal PC matrix} $M=[m_{ij}]$ if every
directed edge $(c_{i},c_{j})\in E_{d}$ implies $m_{ji}=1$ and $m_{ij}=-1$,
and every undirected edge $\{c_{i},c_{j}\}\in E_{u}$ implies $m_{ij}=0$. 
\end{definition}

\begin{definition}
\label{def:triad-and-triad-covering}All three mutually distinct vertices
$t=\{c_{i},c_{k},c_{j}\}\subseteq V$ are said to be a triad. The
vertex $c$ is said to be contained by a triad $t=\{c_{i},c_{k},c_{j}\}$
if $c\in t$. A triad $t=\{c_{i},c_{k},c_{j}\}$ is said to be covered
by the edge $(p,q)\in E_{d}$ if $p,q\in t$.
\end{definition}

Sometimes it will be more convenient to write a triad $t=\{c_{i},c_{k},c_{j}\}$
as the set of edges, e.g. $\{c_{i}\rightarrow c_{k},c_{k}\,\text{---}\,c_{j},c_{i}\,\text{---}\,c_{j}\}$.
However, both notations are equivalent, the latter one allows the
reader to immediately identify the type of triad. 
\begin{definition}

\begin{figure}
\begin{centering}
\includegraphics{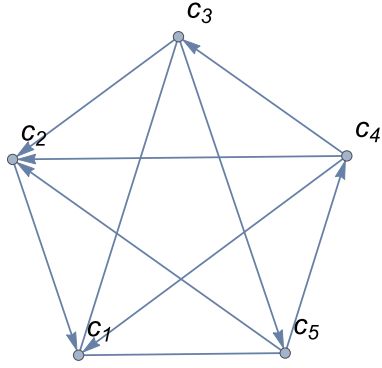}
\par\end{centering}
\caption{The \emph{gt-graph} corresponding to the matrix $M$, see (\ref{eq:A_example}).}
\label{fig:gt-graph-example}
\end{figure}
\end{definition}

In their work, \emph{Kendall and Babington Smith} dealt with the ordinal
pairwise comparisons without ties \cite{Kendall1940otmo}. Hence,
in fact, they do not consider the situation in which $c_{i}\sim c_{j}$.
For the same reason, their \emph{ordinal PC matrices} had no zeros
anywhere outside the diagonal\footnote{In fact, those matrices had no zeros as the authors inserted dashes
on the diagonal \cite{Kendall1940otmo}.}. For the purpose of defining the notion of inconsistency in preferences,
they adopt the transitivity of the preference relationship. According
to this assumption, every triad $c_{i},c_{k},c_{j}$ of three different
alternatives can be classified as consistent or inconsistent (contradictory).
Providing that there are no ties between alternatives, there are two
different kinds of triads (it is easy to verify that any other triad
can be simply boiled down to one of these two by simple index changing).
The first one $c_{i}\rightarrow c_{k},c_{k}\rightarrow c_{j}$ and
$c_{i}\rightarrow c_{j}$ hereinafter referred to as the consistent
triad\footnote{Index $3$ means that this kind of triad is formed by three directed
edges.} $\textit{CT}_{3}$, and $c_{i}\rightarrow c_{k},c_{k}\rightarrow c_{j}$
and $c_{j}\rightarrow c_{i}$ termed hereinafter as the inconsistent
triad $\textit{IT}_{3}$ (Fig. \ref{fig:CT3andIT3}). 

\begin{figure}
\begin{centering}
\subfloat[$\textit{CT}_{3}$ - a consistent triad covered by three directed
edges]{\begin{centering}
\includegraphics{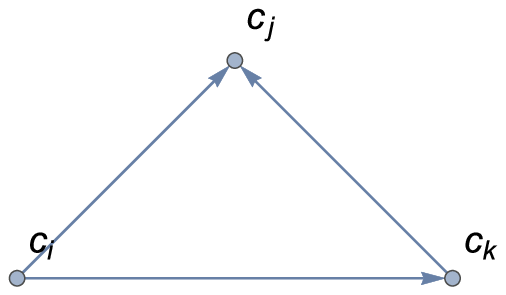}
\par\end{centering}

}~~~~~\subfloat[$\textit{IT}_{3}$ - an inconsistent (circular) triad covered by three
directed edges]{\begin{centering}
\includegraphics{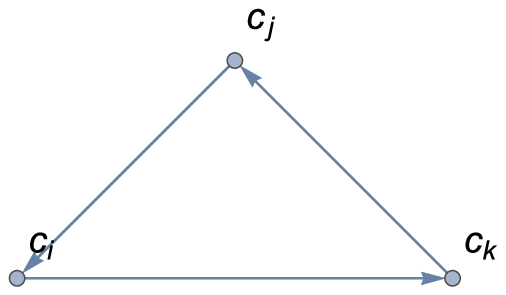}
\par\end{centering}

}
\par\end{centering}
\caption{Triads for paired comparisons without ties}
\label{fig:CT3andIT3}
\end{figure}

Of course, the more inconsistent the triads in the \emph{ordinal}
\emph{PC matrix}, the more inconsistent the set of preferences, hence
the less reliable the conclusions drawn from the set of paired comparisons.
To determine how inconsistent the given set of paired comparisons
is, \emph{Kendall and Babington} \emph{Smith} \cite{Kendall1940otmo}
provide the maximal number of inconsistent triads in the $n\times n$
\emph{PC matrix} without ties. Denoting the actual number of inconsistent
triads in $T_{M}$ by $\left|T_{M}\right|_{i}$, and the maximal possible
number of inconsistent triads in $n\times n$ PC matrix $M$ as $\mathcal{I}(n)$,
we have \footnote{As every $n\times n$ \emph{ordinal PC matrix} $M$ corresponds to
some tournament graph $T_{n}^{\ast}$ we also use the notation $\left|T_{n}^{\ast}\right|_{i}$
to express the number of inconsistent triads in it.}: 
\begin{equation}
\mathcal{I}(n)=\left\{ \begin{array}{ccc}
\frac{n^{3}-n}{24} & \text{when} & \text{n is odd}\\
\frac{n^{3}-4n}{24} & \text{when} & \text{n is even}
\end{array}\right.\label{eq:no_of_inconsistent_triads}
\end{equation}
Therefore, the inconsistency index for $M$ defined in \cite{Kendall1940otmo}
is:
\begin{equation}
\zeta(M)=1-\frac{\left|T_{M}\right|_{i}}{\mathcal{I}(n)}\label{eq:zeta-index}
\end{equation}

Unfortunately, including ties into consideration significantly complicates
the scene. Besides the two types of triads $\textit{CT}_{3}$ and
$\textit{IT}_{3}$ we need to take into consideration an additional
five: 
\begin{lyxlist}{00.00.0000}
\item [{$\textit{CT}_{0}$}] - consistent triad of three equally preferred
alternatives $c_{i},c_{k}$ and $c_{j}$ such that $c_{i}\sim c_{k},c_{k}\sim c_{j}$
and $c_{i}\sim c_{j}$. 
\item [{$\textit{IT}_{1}$}] - inconsistent triad composed of three alternatives
$c_{i},c_{k}$ and $c_{j}$ such that $c_{i}\sim c_{k},c_{k}\sim c_{j}$
and $c_{i}\prec c_{j}$. 
\item [{$\textit{IT}_{2}$}] - inconsistent triad composed of three alternatives
$c_{i},c_{k}$ and $c_{j}$ such that $c_{i}\sim c_{k},c_{k}\prec c_{j}$
and $c_{j}\prec c_{i}$. 
\item [{$\textit{CT}_{2a}$}] - consistent triad composed of three alternatives
$c_{i},c_{k}$ and $c_{j}$ such that $c_{i}\sim c_{k},c_{k}\prec c_{j}$
and $c_{i}\prec c_{j}$. 
\item [{$\textit{CT}_{2b}$}] - consistent triad composed of three alternatives
$c_{i},c_{k}$ and $c_{j}$ such that $c_{i}\sim c_{k},c_{j}\prec c_{k}$
and $c_{j}\prec c_{i}$. 
\end{lyxlist}
The above triads can be easily represented as tournament graphs with
ties (Fig. \ref{fig:triads-specific-for-ties}). With the increased
number of different types of triads in a graph, the maximum number
of inconsistent triads also increases. For example, according to (\ref{eq:no_of_inconsistent_triads})
the maximum number of inconsistent triads in $\mathcal{I}(4)$ without
ties is $2$. When ties are allowed, the maximal number of inconsistent
triads increases to $4$, which is the total number of triads in every
simple graph (i.e. with only one edge between one pair of vertices)
with four vertices. 

\begin{figure}[h]
\begin{centering}
\includegraphics{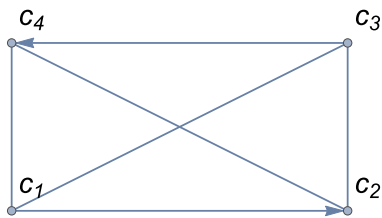}
\par\end{centering}
\caption{$\mathcal{I}(4)$ with four $\textit{IT}_{1}$ triads}
\label{fig:T4with6IT1}
\end{figure}

Let us analyze the graph in (Fig \ref{fig:T4with6IT1}). It is easy
to notice that it contains four $\textit{IT}_{1}$ triads which are:
\{$c_{1}\rightarrow c_{2},\,c_{2}\,\text{---}\,c_{3},\,c_{3}\,\text{---}\,c_{1}$\},
\{$c_{1}\rightarrow c_{2},\,c_{2}\,\text{---}\,c_{4},\,c_{4}\,\text{---}\,c_{1}$\},
\{$c_{1}\text{---}c_{3},\,c_{3}\,\rightarrow\,c_{4},\,c_{4}\,\text{---}\,c_{1}$\},
and \{$c_{2}\,\text{---}\,c_{3},\,c_{3}\,\text{\ensuremath{\rightarrow}}\,c_{4},\,c_{4}\,\text{---}\,c_{1}$\}.
Thus, it is clear that the formulae (\ref{eq:no_of_inconsistent_triads})
and (\ref{eq:zeta-index}) cannot be used to estimate inconsistency
in preferences when ties are allowed. The desire to extend those concepts
to paired comparisons with ties was the main motivation for writing
the work.

\begin{figure}[h]
\begin{centering}
\subfloat[$\textit{CT}_{0}$ - a consistent triad not covered by any directed
edge]{\begin{centering}
\includegraphics{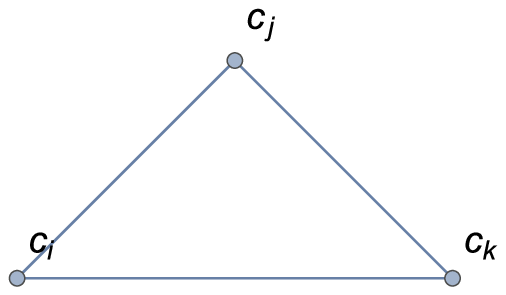}
\par\end{centering}

}~~~~\subfloat[$\textit{IT}_{1}$ - an inconsistent triad covered by one directed
edge]{\begin{centering}
\includegraphics{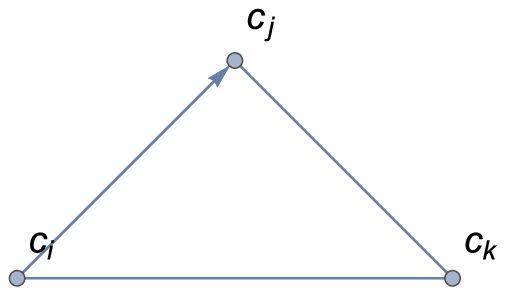}
\par\end{centering}
}
\par\end{centering}
\begin{centering}
\subfloat[$\textit{IT}_{2}$ - an inconsistent triad covered by two directed
edges]{\begin{centering}
\includegraphics{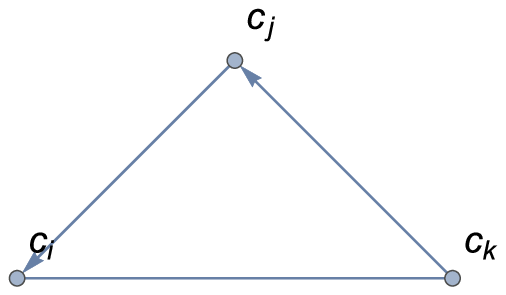}
\par\end{centering}
}~~~~\subfloat[$\textit{CT}_{\textit{2a}}$ - a consistent triad covered by two directed
edges. One alternative is more preferred than two others.]{\begin{centering}
\includegraphics{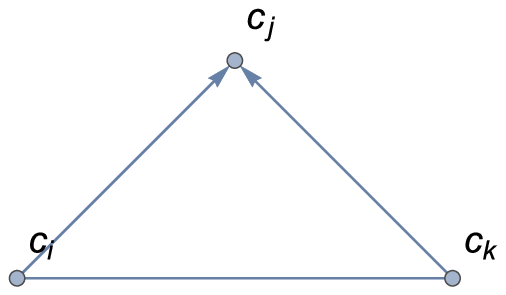}
\par\end{centering}
}
\par\end{centering}
\begin{centering}
\subfloat[$\textit{CT}_{\textit{2b}}$ - a consistent triad covered by two directed
edges. One alternative is less preferred than two others.]{\begin{centering}
\includegraphics{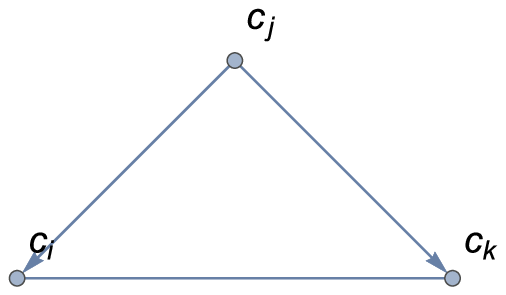}
\par\end{centering}
}
\par\end{centering}
\caption{Triads specific for the pairwise comparisons with ties}
\label{fig:triads-specific-for-ties}
\end{figure}

\section{\label{sec:The-most-inconsistent}The most inconsistent set of preferences
without ties}

To construct the most inconsistent set of pairwise preferences without
ties, let us introduce a few definitions relating to the degree of
vertices. Since every \emph{t-graph} is also a \emph{gt-graph} the
definitions are formulated for the \emph{gt-graph}. 
\begin{definition}
\label{def:The-undirected-degree}Let $G=(V,E_{u},E_{d})$ be a gt-graph
and $c,d\in V$. Then input degree, output degree, undirected degree
and degree of a vertex $c$ are defined as follows: $\text{deg}_{\textit{in}}(c)\overset{\textit{df}}{=}\left|\{d\in V\,:\,d\rightarrow c\in E_{d}\}\right|$,
$\text{deg}_{out}(c)\overset{\textit{df}}{=}\left|\{d\in V\,:\,c\rightarrow d\in E_{d}\}\right|$,
$\text{deg}_{un}(c)\overset{\textit{df}}{=}\left|\{d\in V\,:\,c\,\text{---}\,d\in E_{u}\}\right|$
and $\text{deg}(c)\overset{\textit{df}}{=}\text{deg}_{\textit{in}}(c)+\text{deg}_{out}(c)+\text{deg}_{un}(c)$. 
\end{definition}

\begin{theorem}
\label{the:enforced-triads}Let $G=(V,E_{u},E_{d})$ from $\mathscr{T}_{n}^{g}$.
Then every vertex $c\in V$, for which $\text{deg}_{\textit{in}}(c)=k$
is contained by at least $\binom{k}{2}$ consistent triads of the
type $\textit{CT}_{2a}$ or $\textit{CT}_{3}$. Those triads are said
to be introduced by $c$.
\end{theorem}

\begin{proof}
Let $c_{1},\ldots,c_{k}\in V$ be the vertices such that the edges
$c_{i}\rightarrow c$ are in $E_{d}$. Since $T$ is a \emph{gt-graph}
with $n$ vertices, then for every $c_{i},c_{j}$ where $i,j=1,\ldots,k$
there must exist an edge $c_{i}\rightarrow c_{j}$, $c_{j}\rightarrow c_{i}$
in $E_{d}$ or $c_{i}\,\text{---}\,c_{j}$ in $E_{u}$. In the first
two cases, the vertices $c_{i},c,c_{j}$ make a consistent triad type
$\textit{CT}_{2a}$, whilst in the latter case the vertices $c_{i},c,c_{j}$
form a consistent triad type $\textit{CT}_{3}$. Since there are $k$
vertices adjacent via the incoming edge to $c$ there are at least
as many different consistent triads containing $c$ as two-element
combinations of $c_{1},\ldots,c_{k}$ i.e. $\binom{k}{2}$. See (Fig.
\ref{fig:consistent_triads_enforced}).
\end{proof}

\begin{figure}[h]
\begin{centering}
\includegraphics{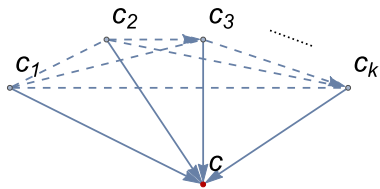}
\par\end{centering}
\caption{Consistent triads introduced by the vertex $c\in V$ with $\text{deg}_{\textit{in}}(c)=k$}
\label{fig:consistent_triads_enforced}
\end{figure}

In general, the given vertex $c$ can form more consistent triads
than those indicated in the above theorem. This is due to the fact
that there may be two or more edges in the form $c\rightarrow c_{k+1},\dots,c\rightarrow c_{k+r}$.
Thus, in $T$ there may also be some number of consistent triads $\textit{CT}_{2b}$
containing $c$. 

The Theorem \ref{the:enforced-triads} is also true for the ordinary
tournament graph (without ties). However, since the only consistent
triads in such a graph are type $\textit{CT}_{3}$ (i.e. there are
no triads of the type $\textit{CT}_{2a}$ or $\textit{CT}_{2b}$ containing
$c$), the only consistent triads containing $c$ are those introduced
by $c$. This leads to the following observation:
\begin{corollary}
Let $T=(V,E_{d})$ from $\mathscr{T}_{n}^{t}$. Then every vertex
$c\in V$, for which $\text{deg}_{in}(c)=k$ is contained by exactly
$\binom{k}{2}$ consistent triads of the type $\textit{CT}_{3}$.
\end{corollary}

Thus, if we would like to construct a tournament graph without ties
which has the maximal number of inconsistent triads, we have to minimize
the number of consistent triads introduced by the vertices, i.e.
\begin{equation}
\left|T\right|_{c}\overset{\textit{df}}{=}\sum_{c\in V}\binom{\text{deg}_{\textit{in}}(c)}{2}\label{eq:no-of-consistent-triads-in-tournament}
\end{equation}
 Since there are no other consistent triads in the tournament graph
than those introduced by the vertices, the expression (\ref{eq:no-of-inconsistent-triads-in-tournament})
denotes, in fact, the number of inconsistent triads in some $T\in\mathscr{T}_{n}^{t}$.
Thus, 
\begin{equation}
\left|T\right|_{i}=\binom{n}{3}-\sum_{c\in V}\binom{\text{deg}_{\textit{in}}(c)}{2}\label{eq:no-of-inconsistent-triads-in-tournament}
\end{equation}

It is commonly known that the sum of degrees in any undirected graph
$G=(V,E)$ equals $2|E|$ \cite[p. 5]{Diestel2005gt}. For the same
reason in $T\in\mathscr{T}_{n}^{t}$ the sum of incoming edges into
vertices is\footnote{Every directed edge corresponds to one victory.}
$|E|=\binom{n}{2}$, i.e.: 
\begin{equation}
\sum_{c\in V}\text{deg}_{\textit{in}}(c)=\binom{n}{2}\label{eq:in-degs-in-tournament}
\end{equation}
Hence, we would like to minimize (\ref{eq:no-of-inconsistent-triads-in-tournament})
providing that the expression (\ref{eq:in-degs-in-tournament}) holds.
Intuitively $\left|T\right|_{i}$ is the largest (\ref{eq:no-of-inconsistent-triads-in-tournament})
i.e. $\left|T\right|_{c}$ is the smallest (\ref{eq:no-of-consistent-triads-in-tournament})
when the input degrees of vertices in a graph are the most evenly
distributed\footnote{As it will be explained latter the input degrees are the most evenly
distributed if for two different vertices $c,d$ holds that $\left|\deg_{\textit{in}}(c)-\deg_{\textit{in}}(d)\right|\leq1$. }. 
\begin{definition}
A \emph{gt-graph} with $n$ vertices is said to be maximal with respect
to the number of inconsistent triads, or briefly \emph{maximal} if
it has the highest possible number of inconsistent triads among the
\emph{gt-graphs} with the size $n$. The fact that the \emph{gt-graph}
is maximal will be denoted $G\in\overline{\mathscr{T}_{n}^{g}}$ or
$T\in\overline{\mathscr{T}_{n}^{t}}$, depending on whether ties are
or are not allowed. $\overline{\mathscr{T}_{n}^{t}}$ and $\overline{\mathscr{T}_{n}^{g}}$
denote families of \emph{gt-graphs} with the highest possible number
of inconsistent triads, i.e. 
\begin{equation}
\overline{\mathscr{T}_{n}^{t}}=\{T\in\mathscr{T}_{n}^{t}\,\textit{such that}\,\left|T\right|_{i}=\underset{T_{r}\in\mathscr{T}_{n}^{t}}{\text{max}}\left|T_{r}\right|_{i}\}\label{eq:familiy-gt-graphs-without-ties}
\end{equation}

\begin{equation}
\overline{\mathscr{T}_{n}^{g}}=\{G\in\mathscr{T}_{n}^{g}\,\textit{such that}\,\left|G\right|_{i}=\underset{G_{r}\in\mathscr{T}_{n}^{g}}{\text{max}}\left|G_{r}\right|_{i}\}\label{eq:familiy-gt-graphs-with-ties}
\end{equation}
\end{definition}

Before we prove the Theorem (\ref{The-number-of}) about the maximal
\emph{t-graph} let us notice that for $r\in\mathbb{N}_{+}$ it holds
that:
\begin{equation}
\binom{2r+1}{2}=r\cdot\left(2r+1\right)\label{eq:podzial-nieparzysty}
\end{equation}

and 
\begin{equation}
\binom{2r}{2}=r\cdot r+r(r-1)\label{eq:podzial-pazysty}
\end{equation}
The above expression (\ref{eq:podzial-nieparzysty}) means that by
adopting $n=2r+1$ as the number of vertices in a graph, we may assign
exactly $r$ incoming edges to every vertex $c$ in $V$ when $n$
is odd. Similarly (\ref{eq:podzial-pazysty}), providing that $n=2r$
is even, we can assign $r$ incoming edges to $r$ vertices and $r-1$
incoming edges to the next $r$ vertices. 
\begin{theorem}
\label{The-number-of}The number of inconsistent triads in the t-graph
$T=(V,E_{d})$ is maximal i.e. $T\in\overline{\mathscr{T}_{n}^{t}}$
if and only if 
\begin{enumerate}
\item for every $c$ in $V$ $\text{deg}_{\textit{in}}(c)=r$ when $n=2r+1$
\item there are $r$ vertices $c_{1},\ldots,c_{r}$ in $V$ such that $\text{deg}_{\textit{in}}(c_{i})=r$,
and $r$ vertices $c_{r+1},\ldots,c_{n}$ such that $\text{deg}_{\textit{in}}(c_{j})=r-1$,
where $n=2r$ and $1\leq i\leq r<j\leq n$.
\end{enumerate}
\end{theorem}

\begin{proofQED}
To prove the theorem, it is enough to show that (\ref{eq:no-of-consistent-triads-in-tournament})
is minimized by the distributions of the vertex degrees mentioned
in the thesis of the theorem. Let us suppose that $n=2r+1$ and (\ref{eq:no-of-consistent-triads-in-tournament})
is minimal but not all the vertices have input degrees equal $r$.
Thus, there must be at least one $c_{i}\in V$ such that $\text{deg}_{\textit{in}}(c_{i})\neq r$.
Let us suppose that $\text{deg}_{\textit{in}}(c_{i})=p>r$ (the second
case is symmetric). Formulae (\ref{eq:in-degs-in-tournament}) and
(\ref{eq:podzial-nieparzysty}) imply that there must also be at least
one $c_{j}\in V$ such that $\text{deg}_{\textit{in}}(c_{j})=q<r$
. Therefore we can decrease $p$ and increase $q$ by one without
changing the sum (\ref{eq:in-degs-in-tournament}) just by replacing
$c_{j}\rightarrow c_{i}$ to $c_{i}\rightarrow c_{j}$. Since $p+q=z$
and $z$ is constant, the sum of consistent triads introduced by $c_{i}$
and $c_{j}$ (Theorem \ref{the:enforced-triads}) is given as: 
\begin{equation}
\binom{p}{2}+\binom{q}{2}=\binom{p}{2}+\binom{z-p}{2}=p(p-z)+\frac{z(z-1)}{2}\label{eq:th-max-nu-triads-in-trn-1}
\end{equation}

Since $z(z-1)/2$ is constant let 
\begin{equation}
f(p)\overset{\textit{df}}{=}p(p-z)+\frac{z(z-1)}{2}\label{eq:th-max-nu-triads-in-trn-2}
\end{equation}

The value $f(p)$ decreases alongside a decreasing $p$ if 
\begin{equation}
f(p)-f(p-1)>0\label{eq:th-max-nu-triads-in-trn-3}
\end{equation}

which is true if and only if 
\begin{equation}
2p>(z-1)\label{eq:th-max-nu-triads-in-trn-4}
\end{equation}

Since $p>q$ and $p+q=z$ the last statement is true, which implies
that, by decreasing $\text{deg}_{\textit{in}}(c_{i})$ and increasing
$\text{deg}_{\textit{in}}(c_{j})$ by one, we can decrease the expression
(\ref{eq:no-of-consistent-triads-in-tournament}). This fact is contrary
to the assumption that (\ref{eq:no-of-consistent-triads-in-tournament})
is minimal, but not all the vertices have input degrees equal $r$. 

The proof for $n=2r$ is analogous to the case when $n=2r+1$ except
the fact that as $c_{i}$ we should adopt such a vertex for which
$\text{deg}_{\textit{in}}(c_{i})\neq r$ and $\text{deg}_{\textit{in}}(c_{i})\neq r-1$.
Note that there must be one if we reject the second statement of the
thesis and, at the same time, we claim that (\ref{eq:no-of-consistent-triads-in-tournament})
is minimal.
\end{proofQED}

The proof of (Theorem \ref{The-number-of}) also suggests an algorithm
that converts any tournament graph into a graph with the maximal number
of inconsistent triads. In every step of such an algorithm, it is
enough to find a vertex $c_{i}$ whose input degree differs from $r$
(when $n$ is odd) or differs from $r$ and $r-1$ (when $n$ is even)
and decreases (or increases) its input degree in parallel with increases
(or decreases) in the input degree of $c_{j}$. If it is impossible
to find such a pair ($c_{i},c_{j}$) this means that the graph is
maximal. The algorithm satisfies the stop condition as with every
iteration the number of inconsistent triads in a graph gets higher
whilst the total number of triads in a graph is bounded by $\binom{n}{3}.$

\emph{Kendall} and \emph{Babington Smith} \cite{Kendall1940otmo}
suggest a way of constructing the most inconsistent graph that brings
to mind \emph{circulant graphs} \cite{Pemmaraju2003cdmc}. Namely,
first add to a graph the cycle $c_{1}\rightarrow c_{2}\rightarrow c_{3}\rightarrow\ldots\rightarrow c_{n}\rightarrow c_{1}$
then the cycle $c_{1}\rightarrow c_{3}\rightarrow c_{5}\rightarrow\ldots\rightarrow c_{n}\rightarrow c_{2}\rightarrow\ldots$
if $n$ is even or two cycles $c_{1}\rightarrow c_{3}\rightarrow\ldots\rightarrow c_{n-1}\rightarrow c_{1}$
and $c_{2}\rightarrow c_{4}\rightarrow\ldots\rightarrow c_{n}\rightarrow c_{2}$
if $n$ is odd, and so on. Adding cycles with more and more skips
needs to be continued until the insertion of all $\binom{n}{2}$ edges.
An example of the maximally inconsistent graphs $T_{X}\in\mathscr{T}_{6}^{t}$
and $T_{Y}\in\mathscr{T}_{7}^{t}$ can be found in (Fig. \ref{fig:most-inconsistent-tournaments}).
Those graphs correspond to the matrices $X$ and $Y$ (\ref{eq:the_most_inconsistent_matrix_tourn}).
\begin{equation}
X=\left(\begin{array}{cccccc}
0 & 1 & 1 & 1 & -1 & -1\\
-1 & 0 & 1 & 1 & 1 & -1\\
-1 & -1 & 0 & 1 & 1 & 1\\
-1 & -1 & -1 & 0 & 1 & 1\\
1 & -1 & -1 & -1 & 0 & 1\\
1 & 1 & -1 & -1 & -1 & 0
\end{array}\right)\,\,\,\,\,Y=\left(\begin{array}{ccccccc}
0 & 1 & 1 & 1 & -1 & -1 & -1\\
-1 & 0 & 1 & 1 & 1 & -1 & -1\\
-1 & -1 & 0 & 1 & 1 & 1 & -1\\
-1 & -1 & -1 & 0 & 1 & 1 & 1\\
1 & -1 & -1 & -1 & 0 & 1 & 1\\
1 & 1 & -1 & -1 & -1 & 0 & 1\\
1 & 1 & 1 & -1 & -1 & -1 & 0
\end{array}\right)\label{eq:the_most_inconsistent_matrix_tourn}
\end{equation}
The Theorem \ref{The-number-of} clearly indicates the form of the
most inconsistent tournament graph, but it does not specify the number
of inconsistent triads in such a graph. This number, however, can
be easily computed using the formula (\ref{eq:no_of_inconsistent_triads}).
To see that the results obtained so far are consistent with (\ref{eq:no_of_inconsistent_triads})
as defined in \cite{Kendall1940otmo} let us prove the following theorem.
\begin{theorem}
\label{theorem:maximal-t-graph-no-inc-triad}For every t-graph $T=(V,E_{d})$
where $T\in\overline{\mathscr{T}_{n}^{t}}$, $n\geq3$ which has the
form defined by the Theorem \ref{The-number-of} it holds that 
\begin{equation}
\left|T\right|_{i}=\mathcal{I}(n)\label{eq:equiv-theorem:1}
\end{equation}
\end{theorem}

\begin{proofQED}
According to (\ref{eq:no-of-inconsistent-triads-in-tournament}) 

\begin{equation}
\left|T\right|_{i}=\binom{2r+1}{3}-\sum_{c\in V}\binom{\text{deg}_{\textit{in}}(c)}{2}\label{eq:equiv-theorem:2}
\end{equation}

Let $n=2r+1$ and $r\in\mathbb{N}_{+}$. Then due to (Theorem \ref{The-number-of}) 

\begin{equation}
\left|T\right|_{i}=\binom{2r+1}{3}-\underset{2r+1}{\left(\underbrace{\binom{r}{2}+\ldots+\binom{r}{2}}\right)}\label{eq:equiv-theorem:3}
\end{equation}

\begin{equation}
\left|T\right|_{i}=\frac{r(2r-1)(2r+1)}{3}-\frac{(r-1)r(2r+1)}{2}\label{eq:equiv-theorem:4}
\end{equation}

\begin{equation}
\left|T\right|_{i}=\frac{r\left(2r^{2}+3r+1\right)}{6}=\frac{(2r+1)^{3}-(2r+1)}{24}\label{eq:equiv-theorem:5}
\end{equation}

\begin{equation}
\left|T\right|_{i}=\frac{(2r+1)^{3}-(2r+1)}{24}=\frac{n^{3}-n}{24}=\mathcal{I}(n)\label{eq:equiv-theorem:6}
\end{equation}

Similarly, when $n=2r$ and $r\in\mathbb{N}_{+}$. Then due to (Th.
\ref{The-number-of}) 

\begin{equation}
\left|T\right|_{i}=\binom{2r}{3}-\underset{r}{\left(\underbrace{\binom{r}{2}+\ldots+\binom{r}{2}}\right)}-\underset{r}{\left(\underbrace{\binom{r-1}{2}+\ldots+\binom{r-1}{2}}\right)}\label{eq:equiv-theorem:7}
\end{equation}

\begin{equation}
\left|T\right|_{i}=\frac{r(2r-2)(2r-1)}{3}-\frac{(r-1)r^{2}}{2}-\frac{(r-2)(r-1)r}{2}\label{eq:equiv-theorem:8}
\end{equation}

\begin{equation}
\left|T\right|_{i}=\frac{r\left(r^{2}-1\right)}{3}=\frac{(2r)^{3}-4(2r)}{24}=\frac{n^{3}-4n}{24}=\mathcal{I}(n)\label{eq:equiv-theorem:9}
\end{equation}

which completes the proof of the theorem.
\end{proofQED}

\begin{figure}
\begin{centering}
\subfloat[$T_{X}\in\overline{\mathscr{T}_{6}^{t}}$]{\begin{centering}
\includegraphics{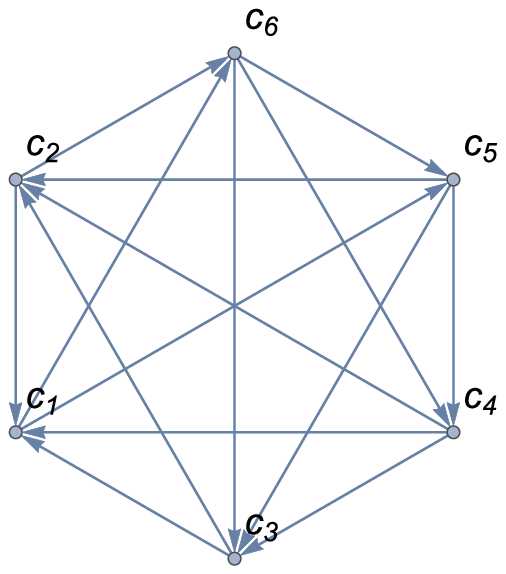}
\par\end{centering}
}~~~~~~\subfloat[$T_{Y}\in\overline{\mathscr{T}_{7}^{t}}$]{\begin{centering}
\includegraphics{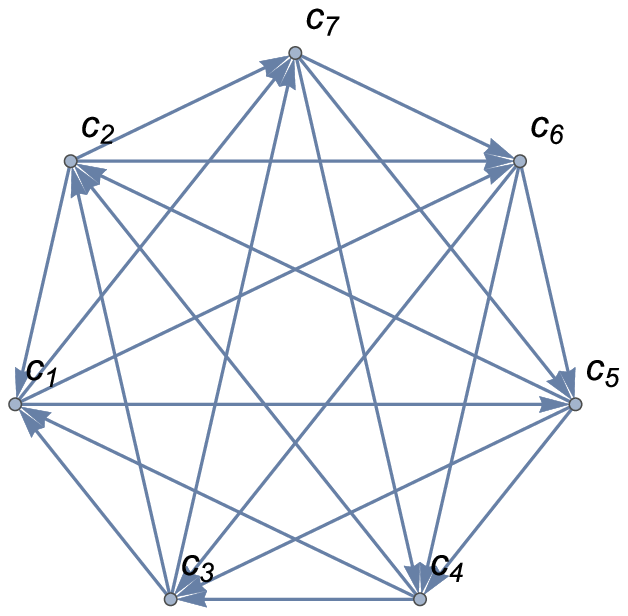}
\par\end{centering}
}
\par\end{centering}
\caption{An example of the most inconsistent tournament graphs with six and
seven vertices}
\label{fig:most-inconsistent-tournaments}
\end{figure}

The above theorem shows that the number of inconsistent triads in
the tournament graph in which input degrees of their vertices are
most evenly distributed is expressed by the formula provided by \emph{Kendall
}and\emph{ Babington Smith} \cite{Kendall1940otmo}. This result,
of course, is the natural consequence of the fact that such a graph
is maximal as regards the number of inconsistent triads, as proven
in (Theorem \ref{The-number-of}).

\section{\label{sec:Properties-of-the}Properties of the most inconsistent
set of preferences with ties}

The graph representation of the set of paired comparisons with ties
is the \emph{gt-graph}. As it may contain two different types of edges,
and hence, essentially more different kinds of triads (Fig. \ref{fig:triads-specific-for-ties}),
the problem of finding the maximum number of inconsistent triads in
such a graph is appropriately more difficult. The reasoning presented
in this section is composed of three parts. In the first part, the
properties of the \emph{gt-graph} are discussed. Next, the maximally
inconsistent \emph{gt-graph} is proposed, and then, we prove that
the proposed graph is indeed maximal with respect to the number of
inconsistent triads. 

The most straightforward example of the fully consistent \emph{gt-graph}
is a complete undirected graph of $n$ vertices (undirected $n$-clique).
It contains only undirected edges, thus all the triads contained in
it are type $\textit{CT}_{0}$. At first glance it seems that by successive
replacing of undirected edges into directed ones we can make the graph
more and more inconsistent. At the beginning, we will try to choose
isolated edges i.e. those which are not adjacent to any directed edge.
It is easy to observe that such edges alone cover $n-2$ different
triads. Hence, by replacing isolated undirected edges into directed
ones we increase the number of inconsistent triads by $n-2$. Unfortunately,
we can insert at most $\left\lfloor \frac{n}{2}\right\rfloor $ isolated
directed edges (every isolated edge needs two vertices out of $n$
only for itself). Then we have to replace not isolated undirected
edges into directed ones, and finally, we decide to make such replacements,
which results in increasing the number of inconsistent triads in a
graph, but also increases input degrees for some vertices. After several
experiments carried out according to the above scheme, one may observe
that it is not easy to choose the edge to replace. However, studying
the above greedy algorithm is not useless. The first thing to notice
is the fact that every \emph{gt-graph} containing more than a certain
number of edges should always have some number of consistent triads.
Another finding is the observation that when constructing a maximal
\emph{gt-graph} one should strive to put at least one directed edge
in each triad. Otherwise, the triad remains consistent, increasing
the chance that the resulting \emph{gt-graph} is not maximal. Both
intuitive observations lead to the conclusion that the construction
of the maximal \emph{gt-graph} is a matter of finding a balance between
too many directed edges resulting in the appearance of consistent
triads of the type $\textit{CT}_{2a}$ and $\textit{CT}_{2b}$ and
too few directed edges resulting in the existence of consistent triads
of the type $\textit{CT}_{0}$. Let us try to formulate this conclusion
in a more formal way. 
\begin{theorem}
\label{theorem:enforced-triads_theorem}Each gt-graph $G\in\mathscr{T}_{n,m}^{g}$
contains at least $\mathcal{C}(n,m)$ consistent triads of the type
$\textit{CT}_{2a}$ or $\textit{CT}_{3}$ where 
\begin{equation}
\mathcal{C}(n,m)=\frac{1}{2}\left\lfloor \frac{m}{n}\right\rfloor \left(2m-n\left\lfloor \frac{m}{n}\right\rfloor -n\right)\label{eq:number-of-enforced-triads}
\end{equation}
\end{theorem}

\begin{proofQED}
The theorem is a straightforward consequence of (Theorem \ref{the:enforced-triads}
and \ref{The-number-of}). The first of them estimates the number
of triads $\textit{CT}_{2a}$ or $\textit{CT}_{3}$ for a given vertex,
whilst the second one shows that the sum of triads $\textit{CT}_{2a}$
or $\textit{CT}_{3}$ introduced by the vertices is minimal when the
input degrees are evenly distributed. As we would like to determine
the lower bound for the number of consistent triads in $G$, we therefore
have to assume that the input degrees are evenly distributed. Since
there are $m$ directed edges in $G$ (it occurs that $m$ times one
alternative is better than the other), then the sum of input degrees
of vertices is $m$. Therefore, adopting an even distribution postulate,
every vertex has at least $\left\lfloor \frac{m}{n}\right\rfloor $
victories assigned (their input degree is at least $\left\lfloor \frac{m}{n}\right\rfloor $).
Of course, the input degree of some of them may be larger by one.
In other words, in the considered \emph{gt-graph} there are $p$ vertices
whose input degree is $\left\lfloor \frac{m}{n}\right\rfloor $ and
$n-p$ vertices whose input degree might be $\left\lfloor \frac{m}{n}\right\rfloor +1$.
According to (Theorem \ref{the:enforced-triads}) such a graph has
at least $\mathcal{C}(n,m)$ consistent triads, where 
\begin{equation}
\mathcal{C}(n,m)=p\binom{\left\lfloor \frac{m}{n}\right\rfloor }{2}+(n-p)\binom{\left\lfloor \frac{m}{n}\right\rfloor +1}{2}\label{eq:th4eq1}
\end{equation}

We know that the sum of input degrees of vertices is $m$, so 
\begin{equation}
p\left\lfloor \frac{m}{n}\right\rfloor +(n-p)\left(\left\lfloor \frac{m}{n}\right\rfloor +1\right)=m\label{eq:th4eq2}
\end{equation}

Hence, 

\begin{equation}
p=n\left(\left\lfloor \frac{m}{n}\right\rfloor +1\right)-m\label{eq:th4eq3}
\end{equation}

Therefore (\ref{eq:th4eq1}) can be written as 
\begin{equation}
\mathcal{C}(n,m)=\left(n\cdot\left(\left\lfloor \frac{m}{n}\right\rfloor +1\right)-m\right)\cdot\binom{\left\lfloor \frac{m}{n}\right\rfloor }{2}+\left(m-n\cdot\left\lfloor \frac{m}{n}\right\rfloor \right)\cdot\binom{\left\lfloor \frac{m}{n}\right\rfloor +1}{2}\label{eq:th4eq4}
\end{equation}

which, after appropriate transformations leads to (\ref{eq:number-of-enforced-triads}). 
\end{proofQED}

The immediate consequence of (Lemma \ref{theorem:enforced-triads_theorem})
is the following corollary: 
\begin{corollary}
\label{corollary:inconsistency-upper-bound}Each \emph{gt-graph} $G\in\mathscr{T}_{n,m}^{g}$
contains at most 
\begin{equation}
\binom{n}{3}-\mathcal{C}(n,m)\label{eq:upper-bound-inconsistent-triads}
\end{equation}

inconsistent triads.

\end{corollary}

For the purpose of further consideration, let us denote by $\mathcal{T}$
a set of all the triads in the \emph{gt-graph} and by $\mathcal{T}_{i}$
- a set of triads covered by $i=0,\ldots,3$ directed edges. For brevity,
we denote the sum $\mathcal{T}_{i}\cup\mathcal{T}_{j}$ as $\mathcal{T}_{i,j}$.
In particular, it holds that $\mathcal{T}=\mathcal{T}_{0}\cup\mathcal{T}_{1}\cup\mathcal{T}_{2,3}$.
This allows the formulation of a quite straightforward but useful
observation. 
\begin{corollary}
\label{corollary:triad_coverage}As every two sets out of $\mathcal{T}_{0},\ldots,\mathcal{T}_{3}$
are mutually disjoint, then for every gt-graph $G\in\mathscr{T}_{n}^{g}$
it is true that

\begin{equation}
\binom{n}{3}=\left|\mathcal{T}_{0}\right|+\left|\mathcal{T}_{1}\right|+\left|\mathcal{T}_{2,3}\right|\label{eq:triads-cover-corollary}
\end{equation}
\end{corollary}

Another important piece of information about the \emph{gt-graph} follows
from the number of undirected edges adjacent to particular vertices.
Such edges may form the triads $\textit{CT}_{0}$ but may also form
the triads $\textit{IT}_{1}$ (Fig. \ref{fig:vertex-where-degun-is-4}).
This observation allows the number of both triad types to be estimated.

\begin{figure}
\begin{centering}
\includegraphics{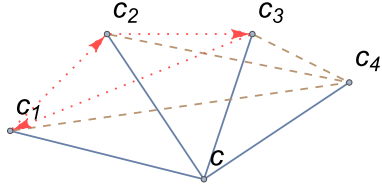}
\par\end{centering}
\caption{Vertex $c$ where $\text{deg}_{\textit{un}}(c)=4$ is contained by
$6$ different triads. Three of them are $\textit{CT}_{0}$ (dashed
edges), the other three are $\textit{IT}_{1}$ (dotted edges). }
\label{fig:vertex-where-degun-is-4}
\end{figure}

\begin{lemma}
\label{lemma:no1}For every \emph{gt-graph} $G\in\mathscr{T}_{n}^{g}$
where $G=(V,E_{u},E_{d})$ it holds that 
\begin{equation}
\sum_{c\in V}\binom{\text{deg}_{\textit{un}}(c)}{2}=3\left|\mathcal{T}_{0}\right|+\left|\mathcal{T}_{1}\right|\label{eq:t0andt1estimationtheorem}
\end{equation}
\end{lemma}

\begin{proofQED}
Let $c_{1}\,\text{---}\,c,\dots,c_{k}\,\text{---}\,c$ be the undirected
edges in $E_{u}$ adjacent to some $c\in V$. There are $\binom{k}{2}$
triads that contain $c$. The type of triad depends on the edge $(c_{i},c_{j})$.
If $(c_{i},c_{j})\in E_{u}$ then the triad belongs to $\mathcal{T}_{0}$
whilst if $(c_{i},c_{j})\in E_{d}$ then the triad is in $\mathcal{T}_{1}$.
While calculating the sum $\sum_{c\in V}\binom{\text{deg}_{\textit{un}}(c)}{2}$
every uncovered triad is counted three times as there are three vertices
adjacent to two undirected edges forming the triad. For the same reason,
the triads covered by one directed edge are taken into account only
once. 
\end{proofQED}

Similarly as before, we try to generalize the result (\ref{eq:t0andt1estimationtheorem})
to all the graphs that have $m$ directed edges. 
\begin{lemma}
\label{lemma:no2}For each gt-graph $G\in\mathscr{T}_{n,m}^{g}$ where
$G=(V,E_{u},E_{d})$ it holds that 
\begin{equation}
\mathcal{D}(n,m)\leq3\left|\mathcal{T}_{0}\right|+\left|\mathcal{T}_{1}\right|\label{eq:0_1_estimation}
\end{equation}

where
\begin{equation}
\mathcal{D}(n,m)=\frac{1}{2}\left(n-\left\lfloor \frac{2m}{n}\right\rfloor -2\right)\left(n^{2}+n\left(\left\lfloor \frac{2m}{n}\right\rfloor -1\right)-4m\right)\label{eq:estim_module}
\end{equation}
\end{lemma}

\begin{proofQED}
Similarly as in (Lemma \ref{theorem:enforced-triads_theorem}) the
left side of (\ref{eq:t0andt1estimationtheorem}) is minimal if undirected
degrees are evenly distributed among the vertices. As for every $c\in V$
it holds that $\text{deg}_{\textit{un}}(c)=\text{deg}(c)-\text{deg}_{in}(c)-\text{deg}_{out}(c)$
then $\text{deg}_{\textit{un}}(c)=n-1-\left(\text{deg}_{in}(c)+\text{deg}_{out}(c)\right)$.
Thus, undirected degrees of vertices are evenly distributed if and
only if the number of directed edges adjacent to the vertices are
evenly distributed. 

It is easy to see that in a \emph{gt-graph} having $m$ directed edges
the sum of input and output degrees is $2m$. Thus, for every graph
that minimizes the left side of (\ref{eq:t0andt1estimationtheorem})
it holds that: 

\begin{equation}
p\left\lfloor \frac{2m}{n}\right\rfloor +(n-p)\left(\left\lfloor \frac{2m}{n}\right\rfloor +1\right)=2m\label{eq:th6_p_estimation}
\end{equation}

The above equality means in particular that in such a graph there
are $p\leq n$ vertices $c_{1},\ldots,c_{p}$ for which $\text{deg}_{\textit{in}}(c_{i})+\text{deg}_{out}(c_{i})=\left\lfloor \frac{2m}{n}\right\rfloor $
and $1\leq i\leq p$, and $n-p$ vertices $c_{p+1},\ldots,c_{n}$
for which $\text{deg}_{\textit{in}}(c_{j})+\text{deg}_{out}(c_{j})=\left\lfloor \frac{2m}{n}\right\rfloor +1$
and $p+1\leq j\leq n$. This statement also implies that in every
graph that minimizes the left side of (\ref{eq:t0andt1estimationtheorem})
there are $p$ vertices $c_{1},\ldots,c_{p}$ for which $\text{deg}_{un}(c_{i})=n-1-\left\lfloor \frac{2m}{n}\right\rfloor $
and $1\leq i\leq p$, and also $n-p$ vertices $c_{p+1},\ldots,c_{n}$
for which $\text{deg}_{un}(c_{j})=n-2-\left\lfloor \frac{2m}{n}\right\rfloor $
and $p+1\leq j\leq n$. 

Thus, for every $G\in\mathscr{T}_{n,m}^{g}$ the lower bound of $3\left|\mathcal{T}_{0}\right|+\left|\mathcal{T}_{1}\right|$
is: 

\begin{equation}
\mathcal{D}(n,m)=p\binom{n-1-\left\lfloor \frac{2m}{n}\right\rfloor }{2}+(n-p)\binom{n-2-\left\lfloor \frac{2m}{n}\right\rfloor }{2}\label{eq:lema_3_eq4}
\end{equation}

Since from (\ref{eq:th6_p_estimation}) $p$ equals
\begin{equation}
p=n\left(\left\lfloor \frac{2m}{n}\right\rfloor +1\right)-2m\label{eq:lema_3_eq5}
\end{equation}

Thus, 

\begin{align}
\mathcal{D}(n,m)= & \left(n\left(\left\lfloor \frac{2m}{n}\right\rfloor +1\right)-2m\right)\binom{n-1-\left\lfloor \frac{2m}{n}\right\rfloor }{2}\nonumber \\
 & +\left(2m-n\left\lfloor \frac{2m}{n}\right\rfloor \right)\binom{n-2-\left\lfloor \frac{2m}{n}\right\rfloor }{2}\label{eq:lema_3_eq6}
\end{align}

The above expression simplifies to 

\begin{equation}
\mathcal{D}(n,m)=\frac{1}{2}\left(-\left\lfloor \frac{2m}{n}\right\rfloor +n-2\right)\left(n\left\lfloor \frac{2m}{n}\right\rfloor -4m+(n-1)n\right)\label{eq:lema_3_eq7}
\end{equation}

which completes the proof of the theorem.
\end{proofQED}

Through the analysis of the degree of vertices we can also estimate
the value $\left|\mathcal{T}_{2,3}\right|$. 
\begin{lemma}
\label{lemma:no3}For every \emph{gt-graph} $G\in\mathscr{T}_{n}^{g}$
where $G=(V,E_{u},E_{d})$ it holds that 
\begin{equation}
\frac{1}{3}\sum_{c\in V}\binom{\text{deg}_{\textit{in}}(c)+\text{deg}_{out}(c)}{2}\leq\left|\mathcal{T}_{2,3}\right|\label{eq:t2andt3estimationtheorem}
\end{equation}
\end{lemma}

\begin{proofQED}
Let $c_{1}\rightarrow c,c\rightarrow c_{2},\dots,c_{k}\rightarrow c$
be the directed edges in $E_{d}$ adjacent to some $c\in V$. There
are $\binom{k}{2}$ triads that contain $c$ where $k=\text{deg}_{\textit{in}}(c)+\text{deg}_{\textit{out}}(c)$,
which are covered by two or three directed edges. While calculating
the sum $\sum_{c\in V}\binom{\text{deg}_{\textit{in}}(c)+\text{deg}_{\textit{out}}(c)}{2}$
triads covered by two directed edges are counted once, whilst all
the triads covered by three directed edges are counted three times.
In the worst case scenario, all the considered triads are covered
by three directed edges. Thus, $\frac{1}{3}\sum_{c\in V}\binom{\text{deg}_{\textit{in}}(c)+\text{deg}_{\textit{out}}(c)}{2}$
is the lower bound for $\left|\mathcal{T}_{2,3}\right|$. This observation
completes the proof.
\end{proofQED}

Similarly as before, let us extend the above Lemma to all \emph{gt-graphs}
that have $n$ vertices and $m$ directed edges. 
\begin{lemma}
\label{lemma:no4}For each gt-graph \textup{$G\in\mathscr{T}_{n,m}^{g}$
where $G=(V,E_{u},E_{d})$} it holds that 
\begin{equation}
\mathcal{E}(n,m)\leq\left|\mathcal{T}_{2,3}\right|\label{eq:t_2_3:eqno:1}
\end{equation}

where
\begin{equation}
\mathcal{E}(n,m)=\frac{1}{6}\left\lfloor \frac{2m}{n}\right\rfloor \left(4m-n\left(\left\lfloor \frac{2m}{n}\right\rfloor +1\right)\right)\label{eq:t_2_3:eqno:2}
\end{equation}
\end{lemma}

\begin{proofQED}
Similarly as in (Lemma \ref{lemma:no2}) the left side of (\ref{eq:t2andt3estimationtheorem})
is minimal if the sum of input and output degrees of the vertices
are evenly distributed. It is easy to see that in a \emph{gt-graph}
that has $m$ directed edges the sum of input and output degrees is
$2m$. Thus, for every graph that minimizes the left side of (\ref{eq:t2andt3estimationtheorem})
it holds that (\ref{eq:th6_p_estimation}). This implies that in the
\emph{gt-graph} which minimizes the left side of (\ref{eq:t2andt3estimationtheorem})
there should be $p$ vertices adjacent to $\left\lfloor \frac{2m}{n}\right\rfloor $
directed edges and $n-p$ vertices adjacent to $\left\lfloor \frac{2m}{n}\right\rfloor +1$
directed edges. Based on (\ref{eq:t2andt3estimationtheorem}) we conclude
that
\begin{equation}
\mathcal{E}(n,m)=\frac{1}{3}\left(p\binom{\left\lfloor \frac{2m}{n}\right\rfloor }{2}+\left(n-p\right)\binom{\left\lfloor \frac{2m}{n}\right\rfloor +1}{2}\right)\label{eq:t_2_3:eqno:3}
\end{equation}

Applying (\ref{eq:lema_3_eq5}) we obtain 

\begin{align}
\mathcal{E}(n,m)= & \frac{1}{3}\left\{ \left[n\left(\left\lfloor \frac{2m}{n}\right\rfloor +1\right)-2m\right]\binom{\left\lfloor \frac{2m}{n}\right\rfloor }{2}\right.\nonumber \\
 & +\left.\left[n-\left(n\left(\left\lfloor \frac{2m}{n}\right\rfloor +1\right)-2m\right)\right]\binom{\left\lfloor \frac{2m}{n}\right\rfloor +1}{2}\right\} \label{eq:t_2_3:eqno:4}
\end{align}

Hence, 

\begin{equation}
\mathcal{E}(n,m)=\frac{1}{3}\left\{ \left(n\left\lfloor \frac{2m}{n}\right\rfloor +n-2m\right)\binom{\left\lfloor \frac{2m}{n}\right\rfloor }{2}\right.\left.+\left(2m-n\left\lfloor \frac{2m}{n}\right\rfloor \right)\binom{\left\lfloor \frac{2m}{n}\right\rfloor +1}{2}\right\} \label{eq:t_2_3:eqno:5}
\end{equation}

The above equation simplifies to 

\begin{equation}
\mathcal{E}(n,m)=\frac{1}{6}\left\lfloor \frac{2m}{n}\right\rfloor \left(4m-n\left\lfloor \frac{2m}{n}\right\rfloor -n\right)\label{eq:t_2_3:eqno:6}
\end{equation}

Which completes the proof of the Lemma.
\end{proofQED}

The Corollary (\ref{corollary:triad_coverage}) and Lemmas (\ref{lemma:no1}
- \ref{lemma:no4}) allow us to estimate the minimal number of consistent
triads which are not covered by any directed edge. 
\begin{theorem}
\label{theorem:empty-triad-lower-bound}For each gt-graph \textup{$G\in\mathscr{T}_{n,m}^{g}$
where $G=(V,E_{u},E_{d})$} holds that
\begin{equation}
\mathcal{F}(n,m)\leq\left|\mathcal{T}_{0}\right|\label{eq:empty-triad-theorem-eq-1}
\end{equation}

where

\begin{equation}
\mathcal{F}(n,m)=\frac{1}{2}\left(\mathcal{D}(n,m)+\mathcal{E}(n,m)-\binom{n}{3}\right)\label{eq:empty-triad-theorem-eq-2}
\end{equation}

which is equivalent to 

\begin{equation}
\mathcal{F}(n,m)=\frac{1}{6}\left(-2n\left\lfloor \frac{2m}{n}\right\rfloor ^{2}+(8m-2n)\left\lfloor \frac{2m}{n}\right\rfloor +(n-2)((n-1)n-6m)\right)\label{eq:empty-triad-theorem-eq-3}
\end{equation}
\end{theorem}

\begin{proofQED}
According to (Corollary \ref{corollary:triad_coverage}) 

\begin{equation}
\binom{n}{3}=\left|\mathcal{T}_{0}\right|+\left|\mathcal{T}_{1}\right|+\left|\mathcal{T}_{2,3}\right|\label{eq:empty-triad-theorem-eq-4}
\end{equation}

Due to (Lemma \ref{lemma:no2}) it holds that
\begin{equation}
\mathcal{D}(n,m)-3\left|\mathcal{T}_{0}\right|\leq\left|\mathcal{T}_{1}\right|\label{eq:empty-triad-theorem-eq-5}
\end{equation}

Therefore it is true that
\begin{equation}
\binom{n}{3}\geq\left|\mathcal{T}_{0}\right|+\left(\mathcal{D}(n,m)-3\left|\mathcal{T}_{0}\right|\right)+\left|\mathcal{T}_{2,3}\right|=\mathcal{D}(n,m)+\left|\mathcal{T}_{2,3}\right|-2\left|\mathcal{T}_{0}\right|\label{eq:empty-triad-theorem-eq-6}
\end{equation}

As we know (Lemma \ref{lemma:no4}) that $\mathcal{E}(n,m)\leq\left|\mathcal{T}_{2,3}\right|$
it is true that
\begin{equation}
\binom{n}{3}\geq\mathcal{D}(n,m)+\mathcal{E}(n,m)-2\left|\mathcal{T}_{0}\right|\label{eq:empty-triad-theorem-eq-7}
\end{equation}

Hence, 
\begin{equation}
\left|\mathcal{T}_{0}\right|\geq\frac{1}{2}\left(\mathcal{D}(n,m)+\mathcal{E}(n,m)-\binom{n}{3}\right)\label{eq:empty-triad-theorem-eq-8}
\end{equation}

which, after simplifying, leads to

\begin{equation}
\left|\mathcal{T}_{0}\right|\geq\frac{1}{6}\left((8m-2n)\left\lfloor \frac{2m}{n}\right\rfloor -2n\left\lfloor \frac{2m}{n}\right\rfloor ^{2}+(n-2)((n-1)n-6m)\right)\label{eq:empty-triad-theorem-eq-9}
\end{equation}

Which completes the proof of the theorem.
\end{proofQED}

One can easily check that for fixed $n$ the values of $\mathcal{F}(n,m)$
decrease to $0$ then become negative, whilst $\left|\mathcal{T}_{0}\right|$
is always a positive integer. Hence, the inequality (\ref{eq:empty-triad-theorem-eq-1})
can also be written as: 
\begin{equation}
\max\{0,\left\lceil \mathcal{F}(n,m)\right\rceil \}\leq\left|\mathcal{T}_{0}\right|\label{eq:ffun-estimation}
\end{equation}

Both theorems \ref{theorem:enforced-triads_theorem} and \ref{theorem:empty-triad-lower-bound}
provide estimations for the minimal number of consistent triads in
a \emph{gt-graph}. Theorem \ref{theorem:enforced-triads_theorem}
provides the lower bound $\mathcal{C}(n,m)$ for the number of triads
$\textit{CT}_{2a}$ and $\textit{CT}_{3}$, whilst Theorem \ref{theorem:empty-triad-lower-bound}
provides the lower bound for the number of consistent triads $\textit{CT}_{0}$.
Hence, the number of consistent triads in the \emph{gt-graph }$T\in\mathscr{T}_{n,m}^{g}$
cannot be lower than $\mathcal{G}(n,m)$ where 
\begin{equation}
\mathcal{G}(n,m)\stackrel{\textit{df}}{=}\mathcal{C}(n,m)+\max\{0,\left\lceil \mathcal{F}(n,m)\right\rceil \}\label{eq:gfun-def}
\end{equation}

Of course, its number could be even higher as we do not care about
triads $\textit{CT}_{2b}$. The immediate consequence of the above
expression is the observation that the number of inconsistent triads
in the \emph{gt-graph} cannot be higher than $\mathcal{H}(n,m)$ where:
\begin{equation}
\mathcal{H}(n,m)\stackrel{\textit{df}}{=}\binom{n}{3}-\mathcal{G}(n,m)\label{eq:hfun-def}
\end{equation}
In particular, the most inconsistent \emph{gt-graph} $G\in\overline{\mathscr{T}_{n}^{g}}$
with some fixed $n\geq3$ can have as many inconsistent triads as
the maximal value of the upper bounding function $\mathcal{H}(n,m)$,
i.e. 
\begin{equation}
\left|G\right|_{i}\leq\max_{0\leq m\leq\binom{n}{2}}\mathcal{H}(n,m)\label{eq:inconsistent-triads-limit-by-hfun}
\end{equation}

Reversely, a \emph{gt-graph} $G\in\mathscr{T}_{n}^{g}$, which fits
that maximum must be maximal i.e. wherever $\left|G\right|_{i}=\max_{0\leq m\leq\binom{n}{2}}\mathcal{H}(n,m)$
then $G\in\overline{\mathscr{T}_{n}^{g}}$. Through the experimental
analysis of the upper bounding function $\mathcal{H}(n,m)$ we can
see that for every fixed $n$ it has one distinct maximum (Fig. \ref{fig:upper-bounding-function-plots}). 

\begin{figure}[h]
\begin{centering}
\subfloat[$\mathcal{H}(n,m)$ for $n=10$ and $m=0,\ldots,\binom{10}{2}$]{\begin{centering}
\includegraphics[scale=0.8]{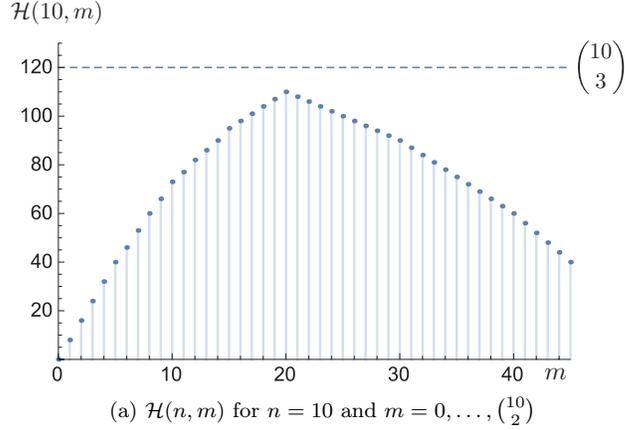}
\par\end{centering}
}
\par\end{centering}
\begin{centering}
\subfloat[$\mathcal{H}(n,m)$ for $n=3,\ldots,20$ and $m=0,\ldots,\binom{n}{2}$]{\begin{centering}
\includegraphics[scale=0.9]{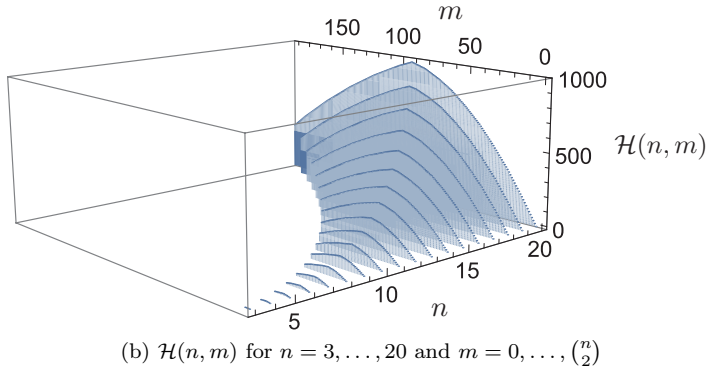}
\par\end{centering}
}
\par\end{centering}
\caption{The upper bounding function $\mathcal{H}(n,m)$}
\label{fig:upper-bounding-function-plots}
\end{figure}

In the next section we propose the graph which fits the maximum of
$\mathcal{H}(n,m)$ and formally prove indispensable theorems. 

\section{\label{sec:The-most-inconsistent-1}The most inconsistent set of
preferences with ties}

In order to find the maximal gt-graph, let us try to look at the function
$\mathcal{H}(n,m)$ and the two functions $\mathcal{C}(n,m)$ and
$\mathcal{F}(n,m)$ of which it is composed (Fig. \ref{fig:bounding_functions}).
$\mathcal{C}(n,m)$ determines the minimal number of consistent triads
covered by more than one directed edge. The more directed edges, the
greater the number of consistent triads in a graph. Hence, for some
small number of directed edges $\mathcal{C}$ equals $0$, then slowly
begins to grow. The function $\mathcal{F}(n,m)$ indicates the minimal
number of triads not covered by any directed edge. Those triads are
also consistent. With the increase in the number of directed edges,
their quantity decreases and eventually reaches $0$. Since for the
positive ordinates $\mathcal{F}$ decreases faster than $\mathcal{C}$
grows, then the function $\mathcal{H}$ reaches the maximum when $\mathcal{F}$
becomes $0$. This indicates that in the optimal \emph{gt-graph} all
the triads should be covered by at least one directed edge. This requires
the introduction of so many directed edges that the number of triads
will become consistent thereby. However, the slope of both functions
$\mathcal{F}$ and $\mathcal{C}$ indicates that it is more important
to cover each triad $\textit{CT}_{0}$ than not to create too many
consistent triads $\textit{CT}_{2a},\textit{CT}_{2b}$ or $\textit{CT}_{3}$. 

\begin{figure}[h]
\begin{centering}
\includegraphics[scale=0.5]{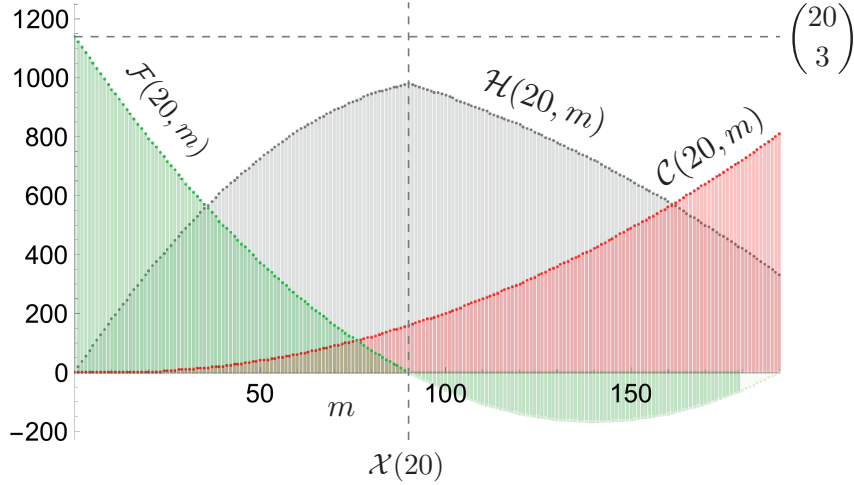}
\par\end{centering}
\caption{Bounding of consistent and inconsistent triads for \emph{gt-graph}
with $n=20$ vertices.}
\label{fig:bounding_functions}
\end{figure}

The considerations in the previous section also indicate that directed
edges should be evenly distributed. Otherwise, the \emph{gt-graph}
may not be maximal. The above somewhat intuitive considerations, based
on the viewing functions in the figure, lead to the definition of
the most inconsistent \emph{gt-graph}. 
\begin{definition}
A double tournament graph (hereinafter referred to as\emph{ dt-graph}),
is a \emph{gt-graph} $G=(V_{1}\cup V_{2},E_{d_{1}}\cup E_{d_{2}},E_{u})$
such that $(V_{1},E_{d_{1}})$ and $(V_{2},E_{d_{2}})$ are \emph{t-graphs},
where $V_{1}\cap V_{2}=\emptyset$ and $E_{u}=\{\{c,d\}\,\,:\,\,c\in V_{1}\,\wedge\,d\in V_{2}\}$.
\end{definition}

It is easy to observe that in every \emph{dt-graph} all triads are
covered by directed edges (Lemma \ref{lemma:None-of-dt-graph-has-uncovered-triad}).
Thus, for every \emph{dt-graph} it holds that $\max\{0,\left\lceil \mathcal{F}(n,m)\right\rceil \}=0$.
This does not guarantee, however, the minimality of $\mathcal{C}(n,m)$.
Let us propose an improved version of the \emph{dt-graph, }which,
as will be shown later, indeed contains the maximal number of inconsistent
triads. 
\begin{proposition}
\label{proposition:maximal-dt-graph-def}The \emph{dt-graph} $T=(V_{1}\cup V_{2},E_{d_{1}}\cup E_{d_{2}},E_{u})$
is the maximal \emph{dt-graph} if $(V_{1},E_{d_{1}})$ and $(V_{2},E_{d_{2}})$
are maximal \emph{t-graphs} where $\left|V_{1}\right|=\left\lfloor \frac{n}{2}\right\rfloor $
and $\left|V_{2}\right|=\left\lceil \frac{n}{2}\right\rceil $. 
\end{proposition}

\begin{figure}[h]
\begin{centering}
\subfloat[$G_{X^{*}}\in\overline{\mathcal{\mathcal{J}}_{6}^{g}}$]{\begin{centering}
\includegraphics{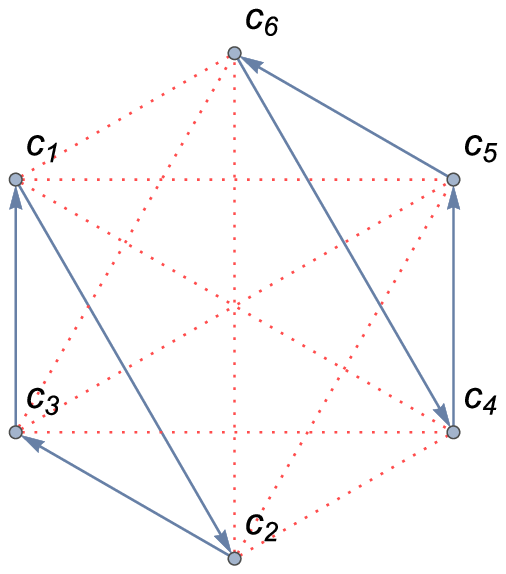}
\par\end{centering}
}~~~~\subfloat[$G_{Y^{*}}\in\overline{\mathcal{\mathcal{J}}_{7}^{g}}$]{\begin{centering}
\includegraphics{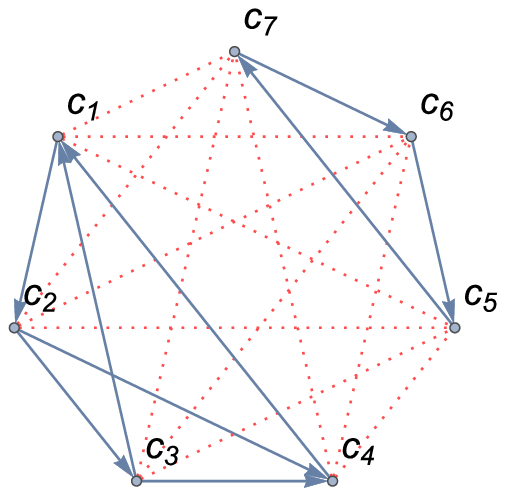}
\par\end{centering}
}
\par\end{centering}
\caption{Two examples of the maximal dt-graphs (undirected edges were dotted). }
\label{fig:maximal-dt-graphs}
\end{figure}

In other words, we suppose that the \emph{dt-graph} with $n$ vertices
composed of two maximal \emph{t-graphs} whose numbers of vertices
are identical (when $n$ is even) or differ by one (when $n$ is odd)
is \emph{maximal. }Examples of such maximal \emph{dt-graph} candidates
can be found at (Fig. \ref{fig:maximal-dt-graphs}). The matrices
that correspond to the graphs $G_{X^{*}}$ and $G_{Y^{*}}$ are given
as (\ref{eq:matrices-of-maximal-dt-graphs}). 

\begin{equation}
X^{*}=\left(\begin{array}{cccccc}
0 & -1 & 1 & 0 & 0 & 0\\
1 & 0 & -1 & 0 & 0 & 0\\
-1 & 1 & 0 & 0 & 0 & 0\\
0 & 0 & 0 & 0 & -1 & 1\\
0 & 0 & 0 & 1 & 0 & -1\\
0 & 0 & 0 & -1 & 1 & 0
\end{array}\right)\,\,\,\,Y^{*}=\left(\begin{array}{ccccccc}
0 & -1 & 1 & 1 & 0 & 0 & 0\\
1 & 0 & -1 & -1 & 0 & 0 & 0\\
-1 & 1 & 0 & -1 & 0 & 0 & 0\\
-1 & 1 & 1 & 0 & 0 & 0 & 0\\
0 & 0 & 0 & 0 & 0 & 1 & -1\\
0 & 0 & 0 & 0 & -1 & 0 & 1\\
0 & 0 & 0 & 0 & 1 & -1 & 0
\end{array}\right)\label{eq:matrices-of-maximal-dt-graphs}
\end{equation}

Let us denote the number of directed edges in a maximal \emph{dt-graph}
candidate by $\mathcal{X}(n)$. It is easy to see that:
\begin{equation}
\mathcal{X}(n)=\binom{\left\lfloor \frac{n}{2}\right\rfloor }{2}+\binom{\left\lceil \frac{n}{2}\right\rceil }{2}\label{eq:no_of_dir_edges_inmax_dtgraph:1}
\end{equation}

\begin{corollary}
\label{corollary:no-of-edges-dt-max-graph}It can be easily calculated
that when $n$ is even i.e. $n=2q$ and $q\in\mathbb{N}_{+}$ it holds
that 
\begin{equation}
\mathcal{X}(2q)=q(q-1)\label{eq:no_of_dir_edges_inmax_dtgraph:2}
\end{equation}

whilst when $n$ is odd i.e. $n=2q+1$ and $q\in\mathbb{N}_{+}$ it
holds that 

\begin{equation}
\mathcal{X}(2q+1)=q^{2}\label{eq:no_of_dir_edges_inmax_dtgraph:3}
\end{equation}
\end{corollary}

To determine the number of consistent/inconsistent triads in this
\emph{``maximal} \emph{gt-graph} \emph{candidate''} let us observe
that all the consistent triads are in the two maximal tournament subgraphs.
This observation can be written in the form of a short Lemma. 
\begin{lemma}
\label{lemma:form-of-a-dt-graph}For every \emph{dt-graph} $G=(V_{1}\cup V_{2},E_{d_{1}}\cup E_{d_{2}},E_{u})$
and a triad $t=\{v_{i},v_{k},v_{j}\}$ if $t\cap V_{1}\neq\emptyset$
and $t\cap V_{2}\neq\emptyset$ then $t$ is inconsistent. 
\end{lemma}

\begin{proofQED}
Since $t\cap V_{1}\neq\emptyset$ and $t\cap V_{2}\neq\emptyset$,
there are two vertices from $t$ in one of the two sets $V_{1}$ and
$V_{2}$ and one vertex from $t$ in the other set. Let us suppose
that $v_{i},v_{k}\in V_{1}$ and $v_{j}\in V_{2}$. Since $(V_{1},E_{d_{1}})$
is a \emph{t-graph} then the edge between $v_{i}$ and $v_{k}$ is
directed. Due to the definition of \emph{dt-graph} both edges $(v_{i},v_{j})$
and $(v_{k},v_{j})$ are undirected, hence $t$ is $IT_{1}$. 
\end{proofQED}

The immediate conclusion can be written as the Lemma
\begin{lemma}
\label{lemma:None-of-dt-graph-has-uncovered-triad}The dt-graph does
not contain uncovered triads
\end{lemma}

\begin{proofQED}
Let us consider the \emph{dt-graph} $G=(V_{1}\cup V_{2},E_{d_{1}}\cup E_{d_{2}},E_{u})$
and a triad $t=\{v_{i},v_{k},v_{j}\}$. If $v_{i},v_{k}\in V_{1}$
and $v_{j}\in V_{2}$ then $t$ is inconsistent (Lemma \ref{lemma:form-of-a-dt-graph}),
hence it cannot be uncovered. If all $v_{i},v_{k},v_{j}\in V_{1}$
then all three edges are spanned between $v_{i},v_{k}$ and $v_{j}$.
Hence, $t$ is covered. The proof is completed as all the other cases
are similar. 
\end{proofQED}

It is also easy to determine the number of inconsistent triads in
the candidate graph. Due to (Theorem \ref{theorem:maximal-t-graph-no-inc-triad})
the number of consistent triads in the maximal tournament sub-graphs
are $\binom{\left\lfloor \frac{n}{2}\right\rfloor }{3}-\mathcal{I}\left(\left\lfloor \frac{n}{2}\right\rfloor \right)$
and $\binom{\left\lceil \frac{n}{2}\right\rceil }{3}-\mathcal{I}\left(\left\lceil \frac{n}{2}\right\rceil \right)$
correspondingly. Since there are no consistent triads in double tournament
graphs, except those that are fully enclosed in the maximal tournament
sub-graphs (Lemma \ref{lemma:form-of-a-dt-graph}), the number of
inconsistent triads in the maximal \emph{gt-graph} candidate is given
as: 
\begin{equation}
\mathcal{Y}(n)=\binom{n}{3}-\left(\binom{\left\lfloor \frac{n}{2}\right\rfloor }{3}-\mathcal{I}\left(\left\lfloor \frac{n}{2}\right\rfloor \right)\right)-\left(\binom{\left\lceil \frac{n}{2}\right\rceil }{3}-\mathcal{I}\left(\left\lceil \frac{n}{2}\right\rceil \right)\right)\label{eq:yfun-def}
\end{equation}
To confirm that a \emph{dt-graph} (Proposition \ref{proposition:maximal-dt-graph-def})
is indeed maximal we need to prove that
\begin{itemize}
\item the function $\mathcal{H}(n,m)$ reaches the maximum when the number
of directed edges in a graph equals $m=\mathcal{X}(n)$, and
\item the maximum of $\mathcal{H}(n,m)$ equals $\mathcal{Y}(n)$
\end{itemize}
Therefore to make the Proposition \ref{proposition:maximal-dt-graph-def}
a fully fledged claim we prove (Theorem \ref{theorem:final-theorem}).
However, before we start (Theorem \ref{theorem:final-theorem}) let
us prove a couple of Lemmas which formally confirm what we have seen
at (Fig. \ref{fig:bounding_functions}). The aim of the first Lemma
(\ref{lemma:f-function}) is a formal confirmation of the shape of
the function $\mathcal{F}$. In particular, it confirms that $\mathcal{F}$
crosses the x-axis at the same point where $\mathcal{H}$ reaches
the maximum i.e. for every fixed $n\geq3$, $\mathcal{F}$ is positive
when $0\leq m<\mathcal{X}(n)$, equals $0$ when $m=\mathcal{X}(n)$
and it is non-positive for $\mathcal{X}(n)\leq m\leq\binom{n}{2}$. 
\begin{lemma}
\label{lemma:f-function}For every $n\in\mathbb{N}_{+},n\geq3$ and
$k\in\mathbb{N}_{+}$ it holds that: 
\begin{equation}
\mathcal{F}(n,\mathcal{X}(n))=0\label{eq:lemma:hfun_bound:1}
\end{equation}

\begin{equation}
\mathcal{F}(n,\mathcal{X}(n)-k)\geq1,\,\,\,\text{where}\,\,\,0<k<\mathcal{X}(n)\label{eq:lemma:hfun_bound:2}
\end{equation}

\begin{equation}
\mathcal{F}(n,\mathcal{X}(n)+k)\leq0,\,\,\,\text{where}\,\,\,0<k\leq\binom{n}{2}-\mathcal{X}(n)\label{eq:lemma:hfun_bound:3}
\end{equation}
\end{lemma}

\begin{proof}
Proof of the Lemma, consisting of elementary but time consuming operations,
can be found in (\ref{sec:Appendix-A}).
\end{proof}

The aim of the next Lemma is to show that $\mathcal{C}$ is strictly
increasing for every $m$ not smaller than $n$ and obviously not
greater than the maximal number of edges in a \emph{gt-graph} i.e.
$\binom{n}{2}$ (Fig. \ref{fig:bounding_functions}). Thus, by adding
more directed edges than $n$ we may only increase the minimal number
of consistent triads of the types $\textit{CT}_{2a}$ or $\textit{CT}_{3}$. 
\begin{lemma}
\label{lemma:c-function}For every $n\in\mathbb{N}_{+},n\geq3$ the
function $\mathcal{C}$ 
\end{lemma}

\begin{enumerate}
\item is constant and equals $\mathcal{C}(n,m)=0$ for every $m$ such that
$0\leq m<n$
\item is strictly increasing for every $m\in\mathbb{N}_{+}$ such that $n\leq m\leq\binom{n}{2}$,
i.e. 
\begin{equation}
\mathcal{C}(n,m+1)-\mathcal{C}(n,m)>0\label{eq:lemma:cfun_increasing}
\end{equation}
\end{enumerate}
\begin{proof}
Proof of the Lemma, consisting of elementary but time consuming operations,
can be found in (\ref{sec:appendix-B}).
\end{proof}

In every gt-graph with $n$ vertices and $m$ directed edges there
are at least $\mathcal{C}(n,m)$ consistent triads $\textit{CT}_{2a}$
or $\textit{CT}_{3}$. This means that in this graph there are at
most $\binom{n}{3}-\mathcal{C}(n,m)$ inconsistent triads. In particular
the Lemma \ref{lemma:c-function-equality} shows that there is no
\emph{gt-graph} with $n$ vertices and $\mathcal{X}(n)$ directed
edges which has more inconsistent triads than the maximal \emph{gt-graph}
defined in (Proposition \ref{proposition:maximal-dt-graph-def}). 
\begin{lemma}
\label{lemma:c-function-equality}For every $n\in\mathbb{N}_{+},n\geq3$
it holds that 
\begin{equation}
\mathcal{Y}(n)=\binom{n}{3}-\mathcal{C}(n,\mathcal{X}(n))\label{eq:lemma:cfun_equality}
\end{equation}
\end{lemma}

\begin{proof}
Proof of the Lemma, composed of elementary but time consuming operations,
can be found in (\ref{sec:appendix:C}).
\end{proof}

The next Lemma shows that the minimal number of consistent triads
in a \emph{gt-graph} decreases along with adding the next directed
edges. Such a decrease continues as long as the number of directed
edges does not reach the value $\mathcal{X}(n)$. In other words,
following the increasing number of directed edges (until there are
less than $\mathcal{X}(n)$ ) the number of inconsistent triads also
increases. 
\begin{lemma}
\label{lemma:g-function}For every $n\in\mathbb{N}_{+},n\geq3$ the
function $\mathcal{G}$ is strictly decreasing for every $m\in\mathbb{N}_{+}$
such that $1\leq m\leq\mathcal{X}(n)$, i.e. 
\begin{equation}
\mathcal{G}(n,m)-\mathcal{G}(n,m+1)>0\,\,\,\textit{where}\,\,\,1\leq m<\mathcal{X}(n)\label{eq:lemma:cfun_decreasing}
\end{equation}
\end{lemma}

\begin{proof}
Proof of the Lemma, composed of elementary but time consuming operations,
can be found in (\ref{sec:appendix:D}).
\end{proof}

For every fixed $n\geq3$ the function $\mathcal{H}$ determines the
maximal possible number of inconsistent triads in every \emph{gt-graph}. 

The aim of the theorem below is to confirm that, indeed, the proposed
\emph{dt-graph} (Proposition \ref{proposition:maximal-dt-graph-def})
is a \emph{maximal gt-graph. }
\begin{theorem}
\label{theorem:final-theorem}For every \emph{dt-graph} $G=(V_{1}\cup V_{2},E_{d_{1}}\cup E_{d_{2}},E_{u})$
with $n$ vertices where $(V_{1},E_{d_{1}})$ and $(V_{2},E_{d_{2}})$
are maximal \emph{t-graphs and} $\left|V_{1}\right|=\left\lfloor \frac{n}{2}\right\rfloor $
and $\left|V_{2}\right|=\left\lceil \frac{n}{2}\right\rceil $ and
$n>3$ it holds that: 
\begin{enumerate}
\item $\mathcal{X}(n)=m$ maximizes $\mathcal{H}(n,m)$, i.e. 
\begin{equation}
\mathcal{H}(n,\mathcal{X}(n))=\max_{0\leq m\leq\binom{n}{2}}\mathcal{H}(n,m)\label{eq:maintheorem:eq:1}
\end{equation}
\item $\mathcal{Y}(n)$ is a maximum of $\mathcal{H}(n,m)$
\begin{equation}
\mathcal{H}(n,\mathcal{X}(n))=\mathcal{Y}(n)\label{eq:maintheorem:eq:2}
\end{equation}
\end{enumerate}
\end{theorem}

\begin{proofQED}
As (\ref{eq:hfun-def}) then the first claim of the theorem is equivalent
to 
\begin{equation}
\mathcal{G}(n,\mathcal{X}(n))=\min_{0\leq m\leq\binom{n}{2}}\mathcal{G}(n,m)\label{eq:maintheorem:eq:3}
\end{equation}
As (\ref{eq:gfun-def}) then the function $\mathcal{G}$ is the sum
of $\mathcal{C}(n,m)$ and $\max\{0,\left\lceil \mathcal{F}(n,m)\right\rceil \}$.
From (Lemma \ref{lemma:c-function}) we know that $\mathcal{C}$ does
not decrease with respect to $m$. On the other hand, due to the (Lemma
\ref{lemma:f-function}) $\mathcal{F}(n,\mathcal{X}(n)+k)\leq0$ for
every $0<k\leq\binom{n}{2}-\mathcal{X}(n)$, which translates to the
observation that for every $m\geq\mathcal{X}(n)$ it holds that $\max\{0,\left\lceil \mathcal{F}(n,m)\right\rceil \}=0$.
Hence, for every $m\geq\mathcal{X}(n)$ the function $\mathcal{G}$
does not decrease and boils down to $\mathcal{G}(n,m)=\mathcal{C}(n,m)$.
In other words 
\begin{equation}
\mathcal{G}(n,\mathcal{X}(n))\leq\mathcal{G}(n,\mathcal{X}(n)+1)\leq\ldots\leq\mathcal{G}(n,\binom{n}{2})\label{eq:maintheorem:eq:4}
\end{equation}
This fact, coupled with (Lemma \ref{lemma:g-function}) i.e. 
\begin{equation}
\mathcal{G}(n,0)>\mathcal{G}(n,1)>\ldots>\mathcal{G}(n,\mathcal{X}(n))\label{eq:maintheorem:eq:5}
\end{equation}
implies that indeed 
\begin{equation}
\mathcal{G}(n,\mathcal{X}(n))=\min_{0\leq m\leq\binom{n}{2}}\mathcal{G}(n,m)\label{eq:maintheorem:eq:6}
\end{equation}
which completes the proof of the first claim (\ref{eq:maintheorem:eq:1})
of the Theorem \ref{theorem:final-theorem}. To prove the second claim
it is enough to recall that for every $m\geq\mathcal{X}(n)$ it holds
that $\mathcal{G}(n,m)=\mathcal{C}(n,m)$. Thus, in particular

\begin{equation}
\mathcal{H}(n,\mathcal{X}(n))=\binom{n}{3}-\mathcal{C}(n,\mathcal{X}(n))\label{eq:maintheorem:eq:7}
\end{equation}
which satisfies the second claim (\ref{eq:maintheorem:eq:2}) of the
Theorem \ref{theorem:final-theorem}, and which thereby confirms the
Proposition \ref{proposition:maximal-dt-graph-def}. 
\end{proofQED}

\section{\label{sec:Inconsistency-indexes-in}Inconsistency indices in paired
comparisons with ties}

As shown in (Section \ref{sec:Model-of-inconsistency}) the inconsistency
index (called there \emph{``coefficient of consistence''}) defined
by \emph{Kendall} and \emph{Babington Smith} \cite[p. 330]{Kendall1940otmo}
cannot be used in the context of ordinal pairwise comparisons with
ties. Thus, in (\ref{eq:zeta-index}) $\mathcal{I}(n)$ needs to be
replaced by $\mathcal{Y}(n)$ - the maximal number of triads in the
case when ties are allowed. The generalized inconsistency index that
covers pairwise comparisons with ties finally takes the form 
\begin{equation}
\zeta_{g}(M)=1-\frac{\left|G_{M}\right|_{i}}{\mathcal{Y}(n)}\label{eq:generalized-kendall-index-1}
\end{equation}
where $M$ is an ordinal \emph{PC} matrix with ties of the size $n\times n$
(Def. \ref{def-the-ordinal-pc-matrix-with-ties}) and G is a gt-graph
corresponding to $M$. The formula (\ref{eq:generalized-kendall-index-1}),
although concise, may not be handy in practice. This is due to the
use in (\ref{eq:yfun-def}) of the floor $\left\lfloor x\right\rfloor $
and ceiling $\left\lceil x\right\rceil $ operations as well as binomial
symbol $\binom{x}{y}$. For this reason, let us simplify (\ref{eq:yfun-def})
depending on whether $n$ and $\nicefrac{n}{2}$ are odd or even.
There are four cases that need to be considered: 
\begin{equation}
\mathcal{Y}(n)=\begin{cases}
\frac{13n^{3}-24n^{2}-16n}{96} & \text{when}\,\,\,n=4q\,\,\,\text{for}\,\,\,q=1,2,3,\ldots\\
\frac{13n^{3}-24n^{2}-19n+30}{96} & \text{when}\,\,\,n=4q+1\,\,\,\text{for}\,\,\,q=1,2,3,\ldots\\
\frac{13n^{3}-24n^{2}-4n}{96} & \text{when}\,\,\,n=4q+2\,\,\,\text{for}\,\,\,q=1,2,3,\ldots\\
\frac{13n^{3}-24n^{2}-19n+18}{96} & \text{when}\,\,\,n=4q+3\,\,\,\text{for}\,\,\,q=0,1,2,\ldots
\end{cases}\label{eq:generalized-kendall-index-2}
\end{equation}

For example, to compute the inconsistency index for the ordinal PC
matrix $M$ (\ref{eq:A_example}) (see Fig. \ref{fig:gt-graph-example})
first it is necessary to compute the number of inconsistent triads
in $M$. Since (\ref{eq:A_example}) has five inconsistent triads:
$(A_{1},A_{2},A_{3})$, $(A_{1},A_{2},A_{5})$, $(A_{1},A_{3},A_{5})$,
$(A_{1},A_{4},A_{5})$ and $(A_{3},A_{4},A_{5})$ then $\left|T_{M}\right|=5$.
On the other hand, $5=4\cdot1+1$ hence, the value $\mathcal{Y}(5)$
is obtained by replacing $n$ with $5$ in the expression $\nicefrac{1}{96}\cdot\left(13n^{3}-24n^{2}-19n+30\right)$,
i.e. $\mathcal{Y}(5)=10$. In other words, in the considered \emph{gt-graph}
(Fig. \ref{fig:gt-graph-example}) five triads out of ten possible
ones are inconsistent. The generalized consistency index for $M$
takes the form:
\begin{equation}
\zeta_{g}(M)=1-\frac{5}{10}=\frac{1}{2}\label{eq:generalized-kendall-index-3}
\end{equation}
Hence the inconsistency level for $M$ (\ref{eq:A_example}) is $50\%$. 

As every \emph{t-graph} is also a \emph{gt-graph} but not reversely
(see Def. \ref{def:t-graph-definition} and \ref{def:gt-graph-definition})
then the generalized inconsistency index $\zeta_{g}$ can also be
used to estimate the inconsistency level of paired comparisons without
ties. Conversely it is not possible. 

Both inconsistency indices $\zeta$ and $\zeta_{g}$ compare the number
of inconsistent triads in $M$ with the maximal number of such triads
in a matrix of the same size as $M$. Hence, for the maximally inconsistent
matrix the index functions will return $1$, whilst the inconsistency
index for a fully consistent matrix is $0$. The maximal value of
the inconsistency index, of course, does not automatically imply that
all the triads in the given matrix are inconsistent. To capture this
phenomenon, let us define the \emph{absolute inconsistency index}
$\eta$ as a ratio of the number of inconsistent triads to the number
of all possible triads in the $n\times n$ matrix $M$. 
\begin{equation}
\eta(M)\stackrel{\textit{df}}{=}\frac{\left|G_{M}\right|_{i}}{\binom{n}{3}}\label{eq:generalized-kendall-index-4}
\end{equation}
Of course, $0\leq\eta(M)\leq1$. If, for example, $\eta(M)=0.4$ then
it would mean that $M$ contains $60\%$ consistent triads and $40\%$
inconsistent triads. The maximal value that $\eta(M)$ may take is
limited by $\mathcal{I}(n)/\binom{n}{3}$ and $\mathcal{Y}(n)/\binom{n}{3}$
for \emph{t-graphs} and \emph{gt-graphs} correspondingly. Thus, for
the larger matrices $\eta(M)$ may never reach $\text{1}$. Let us
consider the first few values of $\mathcal{I}(n)/\binom{n}{3}$ and
$\mathcal{Y}(n)/\binom{n}{3}$ (Fig. \ref{fig:iy-percentage}).

\begin{figure}
\begin{centering}
\includegraphics{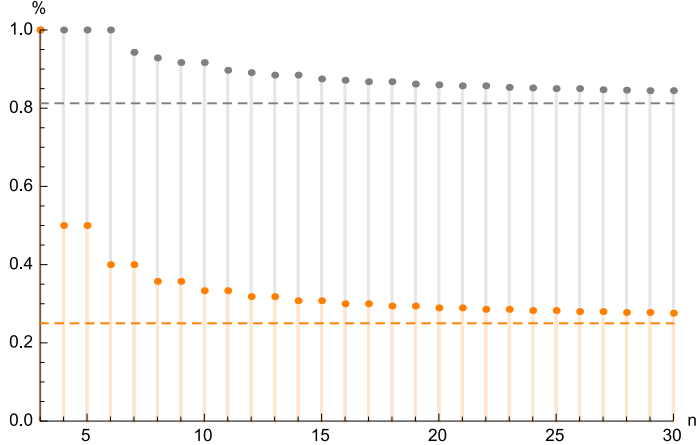}
\par\end{centering}
\caption{The maximal values of $\eta(M)$ for \emph{t-graph} and \emph{gt-graph}}

\label{fig:iy-percentage}
\end{figure}
We can see that for small graphs the percentage of inconsistent triads
is higher than for the larger graphs. In particular, for $n=3,\ldots,6$
there are such \emph{gt-graphs} that have all triads inconsistent.
However, there is only one \emph{t-graph} which has all triads inconsistent.
It is just a single triad. Although the percentage of inconsistent
triads for both \emph{t-graph} and \emph{gt-graph} decrease, they
seem to never drop below certain values. It is easy to compute that\footnote{Expression $\lim_{n\rightarrow\infty}\mathcal{I}(n)/\binom{n}{3}=0.25$
means that both $\lim_{n\rightarrow\infty}\left(\frac{n^{3}-n}{24}\right)/\binom{n}{3}=\lim_{n\rightarrow\infty}\left(\frac{n^{3}-4n}{24}\right)/\binom{n}{3}=0.25$.
Similarly $\lim_{n\rightarrow\infty}\frac{\mathcal{Y}(n)}{\binom{n}{3}}=0.8125$
means that all four limits (see \ref{eq:generalized-kendall-index-2})
equal $0.8125$. }:
\begin{equation}
\lim_{n\rightarrow\infty}\frac{\mathcal{I}(n)}{\binom{n}{3}}=0.25\,\,\,\,\,\text{and}\,\,\,\,\,\lim_{n\rightarrow\infty}\frac{\mathcal{Y}(n)}{\binom{n}{3}}=0.8125\label{eq:generalized-kendall-index-5}
\end{equation}
In other words, although in the larger \emph{t-graphs} $(n>3)$ and
\emph{gt-graphs} $(n>6)$, there must always be consistent triads.
Hence, it is impossible to create a completely inconsistent set of
paired comparisons when the alternatives are more than $3$ (without
ties) and $6$ (when ties are allowed). As we can see very often,
consistent triads must exist. However, it should be remembered that
the ``guaranteed'' number of consistent triads is limited. The expression
(\ref{eq:generalized-kendall-index-5}) implies that at most $75\%$
of triads are ``guaranteed'' to be consistent without ties, and
at most $18.75\%$ of triads are ``guaranteed'' to be consistent
when ties are allowed. 

Figuratively speaking, the possibility of a tie allows us to be much
more inconsistent. However, we rarely have a chance to be completely
inconsistent - only when there are \emph{``sufficiently few''} alternatives.
Fortunately, there is no limit to the number of consistent triads
in a \emph{gt-graph}. Hence, we can be as consistent (and as frequently)
in our views as we want.

\section{\label{sec:Discussion-and-remarks}Discussion and remarks}

To calculate the inconsistency index $\zeta$ or the generalized inconsistency
index $\zeta_{g}$ for some ordinal \emph{PC} $M$ $n\times n$ matrix
we need to determine the number of inconsistent triads in $M$. The
most straightforward method is to consider every single triad and
decide whether it is consistent or not. Since in every complete set
of paired comparisons for $n$ alternatives there are $\binom{n}{3}=\frac{n(n-1)(n-2)}{3}$
different triads, then the running time of such a procedure is $O(n^{3})$.
For \emph{t-graphs}, however, there is a faster way to determine the
number of inconsistent triads in a graph. As mentioned earlier, (\ref{eq:no-of-inconsistent-triads-in-tournament})
denotes the number of inconsistent triads $\left|T\right|_{i}$ in
some \emph{t-graph} $T=(V,E_{d})$. To compute (\ref{eq:no-of-inconsistent-triads-in-tournament})
$\left|T\right|_{i}$ we need to visit every vertex $c\in V$ and
determine its input degree. Computing $\text{deg}_{in}(c)$ for every
$c\in V$ requires visiting every edge $(c_{i},c_{j})\in E_{d}$ twice.
The first time when calculating $\text{deg}_{in}(c_{i})$, the second
time when $\text{deg}_{in}(c_{j})$ is calculated. Thus, determining
$\text{deg}_{in}(c_{1}),\ldots,\text{deg}_{in}(c_{n})$ requires $2\left|E_{d}\right|$
operations. As $\left|E_{d}\right|=\frac{n(n-1)}{2}$ then the actual
running time of computation for (\ref{eq:no-of-inconsistent-triads-in-tournament})
is $O(n(n-1))=O(n^{2})$. For this reason the inconsistency index
$\zeta$ can be determined faster than $\zeta_{g}$. 

Looking at the different types of triads occurring in a \emph{gt-graph}
(Fig. \ref{fig:triads-specific-for-ties}), one may notice that a
triad not covered by any directed edge is consistent, whilst a triad
covered by one directed edge is always inconsistent (see Def. \ref{def:triad-and-triad-covering}).
Therefore the question arises as to whether it is possible to cover
all triads by one directed edge. If not, what is the minimal number
of directed edges covering all triads? Let us try to formally address
this question. Denote the set of directed edges of some \emph{gt-graph}
by $E_{d}=\{(c_{1},c_{2}),(c_{1},c_{3}),\ldots,(c_{n-1},c_{n})\}$
and the set of triads by $\mathcal{T}=\{\{c_{1},c_{2},c_{3}\},\{c_{1},c_{2},c_{4}\},\ldots,\{c_{n-2},c_{n-1},c_{n}\}\}$.
Of course, $\left|E_{d}\right|=\binom{n}{2}$ and $\left|\mathcal{T}\right|=\binom{n}{3}$.
Then, let $B=(V,E)$ be a bipartite graph such that $V=E_{d}\cup\mathcal{T}$
and $E=\{(e,t)\,\,|\,\,(e,t)\in E_{d}\times\mathcal{T}\,\,\text{and}\,\,e\,\,\text{covers}\,\,t\}$.
Hence, we would like to select the minimal subset of edges from $E_{d}$
whose elements cover (i.e. are connected to) every triad in $\mathcal{T}$. 

Let us consider the problem for $n=5$ (Fig. \ref{fig:set-triad-cover-problem-a}). 

\begin{figure}[h]
\begin{centering}
\subfloat[Bipartite graph corresponding to triads cover problem where~$n=5$]{\begin{centering}
\includegraphics[scale=0.8]{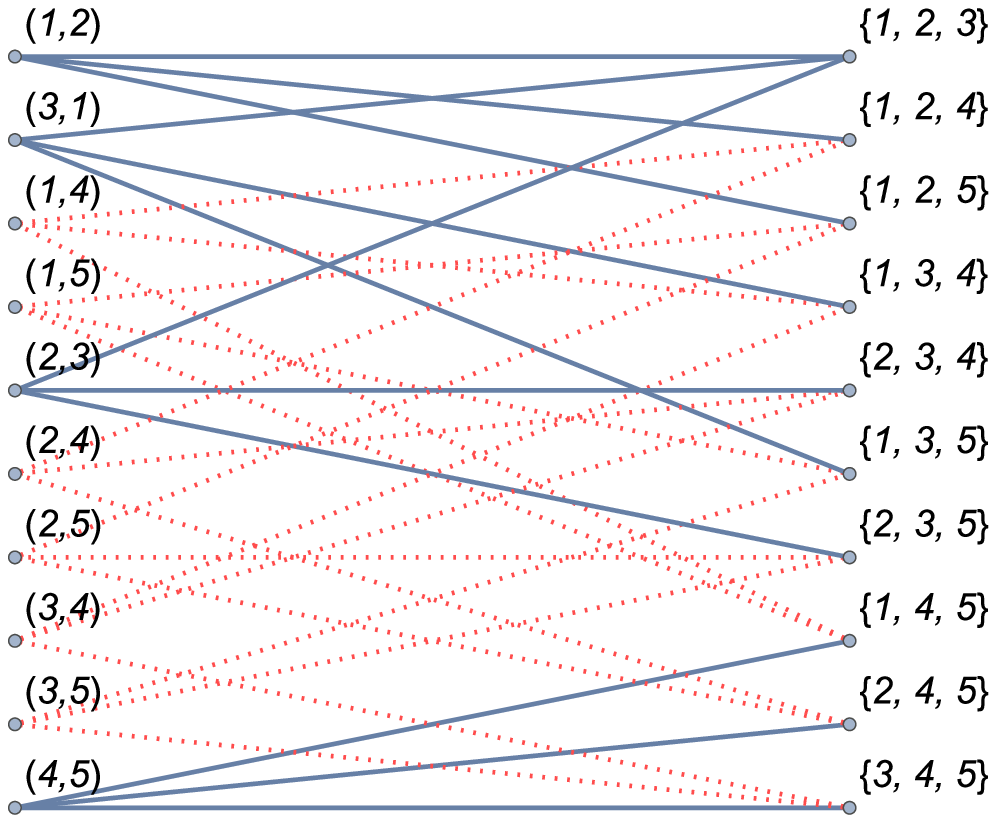}
\par\end{centering}
\label{fig:set-triad-cover-problem-a}

}~~\subfloat[Maximal dt-graph with $5$ vertices (undirected edges are dotted) ]{\begin{centering}
\includegraphics[scale=0.8]{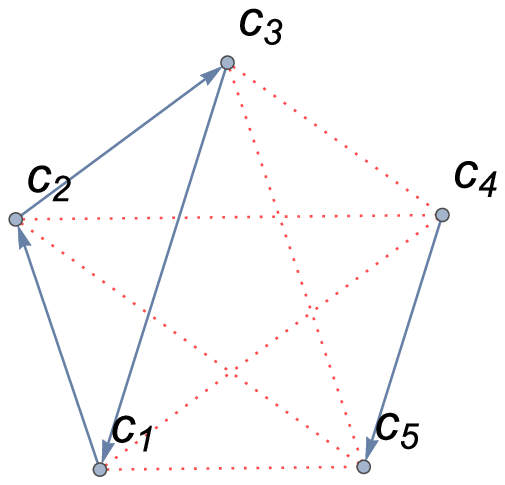}
\par\end{centering}
\label{fig:set-triad-cover-problem-b}}
\par\end{centering}
\caption{Triads cover problem}
\label{fig:set-triad-cover-problem}
\end{figure}

In such a case $E_{d}=\{(1,2)$, $(1,3)$, $(1,4)$, $(1,5)$, $(2,3)$,
$(2,4)$, $(2,5)$, $(3,4)$, $(3,5)$, $(4,5)\}$ and $\mathcal{T}=\{\{1,2,3\}$,
$\{1,2,4\}$, $\{1,2,5\}$, $\{1,3,4\}$, $\{1,3,5\}$, $\{2,3,4\}$,
$\{2,3,5\}$, $\{1,4,5\}$, $\{2,4,5\}$, $\{3,4,5\}\}$. As every
edge covers three different triads we may form the set $S=\{\{t_{i},t_{j},t_{k}\}\,|$
$\,t,t_{j},t_{k}\in\mathcal{T},$ $\exists e\in E_{d}\,\text{that covers}\,\,t_{i},t_{j},t_{k}\}$.
For example, a tripleton $\{\{1,2,3\},\{1,2,4\},\{1,2,5\}\}$ is an
element of $S$ as all its elements are covered by edges $(1,2)$
etc. Thus, the question about the minimal subset of $\left|E_{d}\right|$
whose elements cover all the elements in $\left|\mathcal{T}\right|$,
can be reformulated as follows: what is the minimal subset of $S$
such that the union of its elements equals $\mathcal{T}$?

In general, we can not provide a satisfactory answer to such a question.
The problem we formulate is called a \emph{set cover problem}\footnote{\emph{Wikipedia may serve as a quick reference: \url{https://en.wikipedia.org/wiki/Set_cover_problem}}}
and is one of \emph{Karp's} 21 \emph{NP-complete} problems formulated
in 1972 \cite{Karp1972racp}. Fortunately, we are not dealing with
a \emph{set cover problem} as such, but with its special instance
that can be called \emph{a ``triads cover problem}''. In the latter
case, a \emph{maximal dt-graph} comes to the rescue (\ref{proposition:maximal-dt-graph-def}).
The number of directed edges in the \emph{maximal} \emph{dt-graph}
is $\mathcal{X}(n)$. Due to (Lemma \ref{lemma:f-function}) we know
that every \emph{gt-graph} that has less than $\mathcal{X}(n)$ directed
edges must contain at least one triad of the type $\textit{CT}_{0}$.
On the other hand, any maximal \emph{dt-graph} does not contain uncovered
triads (Lemma \ref{lemma:None-of-dt-graph-has-uncovered-triad}).
This means that a \emph{maximal dt-graph} is a minimal graph covering
all triads by directed edges. 

Let us consider the maximal \emph{dt-graph} for $n=5$. According
to (Proposition \ref{proposition:maximal-dt-graph-def}), such a graph
should be composed of two maximal subgraphs having $\left\lfloor \frac{5}{2}\right\rfloor =3$
and $\left\lceil \frac{5}{2}\right\rceil =2$ vertices. An instance
of the first subgraph can be a triad $(c_{1},c_{2}),\,(c_{2},c_{3})$
and $(c_{3},c_{1})$ whilst the second subgraph is just a single edge
$(c_{4},c_{5})$. As the maximal \emph{dt-graph }with $5$ vertices
provides a minimal edge covering of triads in $5$-clique then the
minimal subset of $S$ that covers the entire $\mathcal{T}$ is, for
example, $\{\{1,2,3\}$, $\{1,2,4\}$, $\{1,2,5\}\}$, $\{\{1,2,3\}$,
$\{1,3,4\}$, $\{1,3,5\}\}$ $\{\{1,2,3\}$, $\{2,3,4\}$, $\{2,3,5\}\}$
and $\{\{1,4,4\}$, $\{2,4,5\}$, $\{3,4,5\}\}$ (Fig. \ref{fig:set-triad-cover-problem}). 

\section{\label{sec:Summary}Summary}

In the presented article, the inconsistency index proposed by \emph{Kendall
and Babington Smith} \cite{Kendall1940otmo} has been extended to
cover pairwise comparisons with ties. For this purpose, the most inconsistent
sets of pairwise comparisons with and without ties have been analyzed.
To model pairwise comparisons with ties a generalized tournament graph
has been defined. An additional \emph{absolute consistency index}
$\eta$ for pairwise comparisons with and without ties has also been
proposed. The relationship between the maximally inconsistent set
of pairwise comparisons with ties and the set cover problem has also
been shown. 

\section*{Acknowledgements }

I would like to thank Prof. Andrzej Bielecki and Dr. Hab. Adam S\k{e}dziwy
for their insightful comments, corrections and reading of the first
version of this work. Special thanks are due to Ian Corkill for his
editorial help. The research is supported by AGH University of Science
and Technology, contract no.: 11.11.120.859. 

\section*{Literature}

\bibliographystyle{plain}
\bibliography{papers_biblio_reviewed}

\appendix
\renewcommand*{\appendixname}{Appendix }

\section{\label{sec:Appendix-A}Proof of Lemma \ref{lemma:f-function}}

\noindent\textsc{Thesis.} 

For every $n\in\mathbb{N}_{+},n\geq3$ and $k\in\mathbb{N}_{+}$ it
holds that: 
\[
\mathcal{F}(n,\mathcal{X}(n))=0\tag{\ref{eq:lemma:hfun_bound:1}}
\]

\[
\mathcal{F}(n,\mathcal{X}(n)-k)\geq1,\,\,\,\text{where}\,\,\,0<k\leq\mathcal{X}(n)\tag{\ref{eq:lemma:hfun_bound:2}}
\]

\[
\mathcal{F}(n,\mathcal{X}(n)+k)\leq0,\,\,\,\textit{where}\,\,\,0<k\leq\binom{n}{2}-\mathcal{X}(n)\tag{\ref{eq:lemma:hfun_bound:3}}
\]

\smallskip\textsc{Proof. Equation (\ref{eq:lemma:hfun_bound:1}), part 1.}\smallskip 

Let $n$ be even i.e. $n=2q$ where $q\in\mathbb{N}_{+}$. Thus, let
us insert to (\ref{eq:empty-triad-theorem-eq-3}) as $n$ the value
$2q$ and as $m$ the value $\mathcal{X}(2q)$. After a series of
elementary transformations applied to (\ref{eq:empty-triad-theorem-eq-2})
we obtain: 
\begin{align}
\mathcal{F}(2q,\mathcal{X}(2q)) & =\frac{1}{3}(-2)q\left(\lfloor q\rfloor^{2}+(1-2q)\lfloor q\rfloor+(q-1)q\right)\label{eq:lemma:hfun_bound:4}
\end{align}

Since $q\in\mathbb{N}_{+}$ then 
\begin{equation}
\lfloor q\rfloor=q\label{eq:lemma:hfun_bound:5}
\end{equation}

Thus, 
\begin{align}
\mathcal{F}(2q,\mathcal{X}(2q)) & =\frac{1}{3}(-2)q\left(q^{2}+(q-1)q+(1-2q)q\right)\label{eq:lemma:hfun_bound:6}
\end{align}

Which after reduction leads to 
\begin{equation}
\mathcal{F}(2q,\mathcal{X}(2q))=0\label{eq:lemma:hfun_bound:7}
\end{equation}

\smallskip\textsc{Proof. Equation (\ref{eq:lemma:hfun_bound:1}), part 2.}\smallskip 

Let $n$ be odd i.e. $n=2q+1$ where $q\in\mathbb{N}_{+}$. Similarly,
let us replace n in (\ref{eq:empty-triad-theorem-eq-3}) by $2q+1$
and $m$ by $\mathcal{X}(2q+1)$. After elementary transformations
we obtain: 

\begin{align}
\mathcal{F}(2q+1,\mathcal{X}(2q+1))= & -\frac{1}{3}(2q+1)\left\lfloor \frac{2q^{2}}{2q+1}\right\rfloor ^{2}\nonumber \\
 & +\frac{1}{3}\left(4q^{2}-2q-1\right)\left\lfloor \frac{2q^{2}}{2q+1}\right\rfloor \nonumber \\
 & +\frac{1}{3}\left(-2q^{2}+3q-1\right)q\label{eq:lemma:hfun_bound:10}
\end{align}

Since $q\in\mathbb{N}_{+}$, we can bound $2q^{2}/\left(2q+1\right)$
from above

\begin{equation}
\frac{2q^{2}}{2q+1}<\frac{2q^{2}}{2q}=q\label{eq:lemma:hfun_bound:11}
\end{equation}

and below

\begin{equation}
q-1=\frac{2\left(q-1\right)^{2}}{2\left(q-1\right)}<\frac{2\left(q-1\right)^{2}}{2q+1}=\frac{2q^{2}-2q+2}{2q+1}\leq\frac{2q^{2}}{2q+1}\label{eq:lemma:hfun_bound:12}
\end{equation}

Therefore, when $q$ is a positive integer it is true that

\begin{equation}
\left\lfloor \frac{2q^{2}}{2q+1}\right\rfloor =(q-1)\label{eq:lemma:hfun_bound:13}
\end{equation}

By applying (\ref{eq:lemma:hfun_bound:13}) to (\ref{eq:lemma:hfun_bound:10})
we obtain 
\begin{align}
\mathcal{F}(2q+1,\mathcal{X}(2q+1)) & =\frac{1}{3}\left(4q^{2}-2q-1\right)(q-1)\nonumber \\
 & +\frac{1}{3}q\left(-2q^{2}+3q-1\right)\nonumber \\
 & -\frac{1}{3}(2q+1)\left(q-1\right)^{2}\label{eq:lemma:hfun_bound:14}
\end{align}

Then, after making further transformations it is easy to verify that:
\begin{equation}
\mathcal{F}(2q+1,\mathcal{X}(2q+1))=0\label{eq:lemma:hfun_bound:15}
\end{equation}

which completes the proof of (\ref{eq:lemma:hfun_bound:1}).

\smallskip\textsc{Proof. Equation (\ref{eq:lemma:hfun_bound:2}), part 1.}\smallskip 

Let $n$ be even i.e. $n=2q$ where $q\in\mathbb{N}_{+}$. Thus, to
prove that $\mathcal{F}(n,\mathcal{X}(n)-k)$ is greater than $0$
it is enough to show that for every $q\geq2$ and $1\leq k<q(q-1)$
it holds that $\mathcal{F}(n,\mathcal{X}(n)-k)>1$. Thus, let us insert
to (\ref{eq:empty-triad-theorem-eq-3}) as $n$ the value $2q$.
After a series of elementary transformations applied to (\ref{eq:empty-triad-theorem-eq-2})
we obtain: 

\begin{equation}
\mathcal{F}(2q,\mathcal{X}(2q)-k)=\frac{2}{3}\left(-q\left\lceil \frac{k}{q}\right\rceil ^{2}+(2k+q)\left\lceil \frac{k}{q}\right\rceil +k(q-1)\right)\label{eq:lemma:hfun_bound:16}
\end{equation}

Let us observe that for the positive integer $p=1,2,\ldots$ if $p\cdot q\leq k<(p+1)q-1$
then $\left\lceil \frac{k}{q}\right\rceil =p$. In order to analyze
$\mathcal{F}$ let us replace $\left\lceil \frac{k}{q}\right\rceil $
by $p$ and define $h$ such that 
\begin{equation}
h(q,k)=\frac{2}{3}\left(p(2k+q)+k(q-1)-qp^{2}\right)\label{eq:lemma:hfun_bound:17}
\end{equation}
where $p\cdot q\leq k<(p+1)q-1$ for every $p=1,2,\ldots,q-2$. Of
course, when $p\cdot q\leq k<(p+1)q-1$ it holds that 
\begin{equation}
\mathcal{F}(2q,\mathcal{X}(2q)-k)=h(q,k)\label{eq:lemma:hfun_bound:17a}
\end{equation}

As $h$ is linear with respect to $k$ then in order to check whether
$h(k)>0$ it is enough to check whether $h$ is greater than $0$
at both ends of the considered interval. So,
\begin{equation}
h(q,p\cdot q)=\frac{2}{3}pq(p+q)\label{eq:lemma:hfun_bound:18}
\end{equation}

and 
\begin{equation}
h(q,(p+1)q-1)=\frac{1}{3}\left(2p^{2}q+2pq^{2}+4pq-4p+2q^{2}-4q+2\right)\label{eq:lemma:hfun_bound:19}
\end{equation}

Since for $p,q=1,2,\ldots$ it holds that $4pq\geq4p$ and $2p^{2}q+2pq^{2}\geq4q$
then 
\begin{equation}
h(q,(p+1)q-1)\geq\frac{1}{3}\left(2q^{2}+2\right)\geq\frac{1}{3}\left(2+2\right)>1\label{eq:lemma:hfun_bound:20}
\end{equation}

Thus, for every $p\cdot q\leq k<(p+1)q-1$ where $p=1,2,\ldots,q-2$,
$h(k)>0$. We just need to check $h$ for $k=q(q-1)$. In such a case
$\left\lceil \frac{k}{q}\right\rceil =q-1$. Thus $h(q(q-1))$ takes
the form:

\begin{equation}
h(q,q(q-1))=\frac{2}{3}q\left(2q^{2}-3q+1\right)\label{eq:lemma:hfun_bound:21}
\end{equation}

As $q\geq2$ then it is easy to verify that $h(q,q(q-1))>0$. 

Since $h(q,k)>0$ for every $p=1,2,\ldots,q-2$, where $p\cdot q\leq k<(p+1)q-1$
and for $k=q(q-1)$ then also $\mathcal{F}(2q,\mathcal{X}(2q)-k)>0$
for $n=2q$ and $1\leq k<q(q-1)$, which completes the first part
of the proof. 

\smallskip\textsc{Proof. Equation (\ref{eq:lemma:hfun_bound:2}), part 2.}\smallskip 

Let $n$ be even i.e. $n=2q+1$ where $q\in\mathbb{N}_{+}$. Thus,
let us insert to (\ref{eq:empty-triad-theorem-eq-3}) as $n$ the
value $2q+1$ and $\mathcal{X}(2q+1)-k$, where this time $1\leq k\leq q^{2}$
(see \ref{eq:no_of_dir_edges_inmax_dtgraph:3}). After a series of
elementary transformations applied to (\ref{eq:empty-triad-theorem-eq-2})
we obtain:

\begin{align}
\mathcal{F}(n,\mathcal{X}(n)-k)= & \frac{1}{3}\left((4k+2q+1)\left\lceil \frac{2\left(k-q^{2}\right)}{2q+1}\right\rceil +4q^{2}\left\lfloor \frac{2\left(q^{2}-k\right)}{2q+1}\right\rfloor -\right.\nonumber \\
 & \left.(2q+1)\left\lfloor \frac{2\left(q^{2}-k\right)}{2q+1}\right\rfloor ^{2}+(2q-1)\left(3k-q^{2}+q\right)\right)\label{eq:lemma:hfun_bound:22}
\end{align}

Since for every $x\in\mathbb{R}$ it holds \footnote{A quick reference is \url{https://en.wikipedia.org/wiki/Floor_and_ceiling_functions}}
\cite{Graham1994cm} that $-\left\lceil x\right\rceil =\left\lfloor -x\right\rfloor $,
and $\mathcal{X}(n)=\mathcal{X}(2q+1)=q^{2}$ then

\begin{align}
\mathcal{F}(2q+1,q^{2}-k)= & \frac{1}{3}\left(-(4k+2q+1)\left\lfloor \frac{2\left(q^{2}-k\right)}{2q+1}\right\rfloor +4q^{2}\left\lfloor \frac{2\left(q^{2}-k\right)}{2q+1}\right\rfloor -\right.\nonumber \\
 & \left.(2q+1)\left\lfloor \frac{2\left(q^{2}-k\right)}{2q+1}\right\rfloor ^{2}+(2q-1)\left(3k-q^{2}+q\right)\right)\label{eq:lemma:hfun_bound:23}
\end{align}

It is easy to observe the relationship between $\left\lfloor \frac{2\left(q^{2}-k\right)}{2q+1}\right\rfloor $
and $k$ is:

$\left\lfloor \frac{2\left(q^{2}-k\right)}{2q+1}\right\rfloor =0$
if and only if $0\leq2\left(q^{2}-k\right)<2q+1$, in other words,
we require that $q^{2}-q-\frac{1}{2}\leq k<q^{2}$ 

$\left\lfloor \frac{2\left(q^{2}-k\right)}{2q+1}\right\rfloor =1$
if and only if $2q+1\leq2\left(q^{2}-k\right)<2(2q+1)$ which translates
to the interval: $\frac{1}{2}\left(2q^{2}-2\left(2q+1\right)\right)\leq k<\frac{1}{2}\left(2q^{2}-1\left(2q+1\right)\right)$

$\left\lfloor \frac{2\left(q^{2}-k\right)}{2q+1}\right\rfloor =2$
if and only if $2(2q+1)\leq2\left(q^{2}-k\right)<3(2q+1)$, hence
$\frac{1}{2}\left(2q^{2}-3\left(2q+1\right)\right)\leq k<\frac{1}{2}\left(2q^{2}-2\left(2q+1\right)\right)$

and in general, $r\overset{\textit{df}}{=}\left\lfloor \frac{2\left(q^{2}-k\right)}{2q+1}\right\rfloor $
if and only if $(r-1)(2q+1)\leq2\left(q^{2}-k\right)<r(2q+1)$, which
translates to the interval for $k$: $\frac{1}{2}\left(2q^{2}-r\left(2q+1\right)\right)\leq k<\frac{1}{2}\left(2q^{2}-(r-1)\left(2q+1\right)\right)$.

Thus, instead of analyzing $\mathcal{F}$ with respect to $k$ over
the whole domain i.e. $1\leq k\leq q^{2}$ and $q\geq2$ we can analyze
it in the subsequent intervals, in which the value $\left\lfloor \frac{2\left(q^{2}-k\right)}{2q+1}\right\rfloor $
is known and fixed. 

Let us introduce the auxiliary function $h:$
\begin{equation}
h(q,k,r)\overset{\textit{df}}{=}\mathcal{F}(2q+1,q^{2}-k)\label{eq:lemma:hfun_bound:24}
\end{equation}
 defined for $k$ such that $\frac{1}{2}\left(2q^{2}-r\left(2q+1\right)\right)\leq k<\frac{1}{2}\left(2q^{2}-(r-1)\left(2q+1\right)\right)$.
Hence, 

\begin{equation}
h(q,k,r)=\frac{1}{3}\left(-(4k+2q+1)r+4q^{2}r-(2q+1)r^{2}+(2q-1)\left(3k-q^{2}+q\right)\right)\label{eq:lemma:hfun_bound:25}
\end{equation}

Moreover, $r$ is the highest when $k$ is $1$. Thus, due to (\ref{eq:lemma:hfun_bound:13})
it holds that $\left\lfloor \frac{2\left(q^{2}-1\right)}{2q+1}\right\rfloor \leq\left\lfloor \frac{2q^{2}}{2q+1}\right\rfloor =q-1$.
Therefore, we know that $r\leq q-1$. Hence, instead of showing that
$\mathcal{F}(2q+1,q^{2}-k)>1$ for every $0\leq k\leq q^{2}$, we
prove that $h(q,k,r)>1$ when $\frac{1}{2}\left(2q^{2}-r\left(2q+1\right)\right)\leq k<\frac{1}{2}\left(2q^{2}-(r-1)\left(2q+1\right)\right)$
for every $0\leq r\leq q-1$. 

Let us observe that $h(q,k,r)$ is a decreasing function with respect
to $k$. That is because
\begin{equation}
h(q,k,r)-h(q,k-1,r)=2q-\frac{4r}{3}+1\label{eq:lemma:hfun_bound:27}
\end{equation}

where $r\leq q-1$. In particular, it is easy to verify that always
$2q+1>\frac{4r}{3}$ for $r\leq q-1$. 

The above equalities justify the following estimation: 
\begin{equation}
h(q,k,r)>h(q,k-1,r)>\ldots>h(q,\frac{1}{2}\left(2q^{2}-r\left(2q+1\right)\right),r)\label{eq:lemma:hfun_bound:28}
\end{equation}

Thus, to prove that $h(q,k,r)>0$ for all admissible values of $q,k,r$
we need to check whether $h(q,\frac{1}{2}\left(2q^{2}-r\left(2q+1\right)\right),r)>0$
for $0\leq r\leq q-1$. 

So, applying the lower bound for $k$, i.e. $k=\frac{1}{2}\left(2q^{2}-r\left(2q+1\right)\right)$
to (\ref{eq:lemma:hfun_bound:25}) we obtain
\begin{equation}
h(q,\frac{1}{2}\left(2q^{2}-r\left(2q+1\right)\right),r)=\frac{1}{6}(2q+1)\left(4q^{2}-6qr-2q+2r^{2}+r\right)\label{eq:lemma:hfun_bound:29}
\end{equation}

Let us denote $h_{2}(q,r)\overset{\textit{df}}{=}h(q,\frac{1}{2}\left(2q^{2}-r\left(2q+1\right)\right),r)$.
It is easy to observe that $h_{2}$ is a parabola with respect to
$r$. Since $\frac{\partial^{2}h_{2}}{\partial r^{2}}=\frac{2}{3}(2q+1)$
is greater than $0$ for $q\geq2$, thus $h_{2}(q,r)$ has the minimum
with respect to $r$ when $\frac{\partial h_{2}}{\partial r}=0$.
I.e. 
\begin{equation}
\frac{\partial h_{2}}{\partial r}=-\frac{1}{6}(2q+1)(6q-4r-1)=0\label{eq:lemma:hfun_bound:30}
\end{equation}

i.e., when 
\begin{equation}
r=\frac{1}{4}(6q-1)\label{eq:lemma:hfun_bound:31}
\end{equation}

In other words, $h_{2}$ decreases for $r=1,2,\ldots$, then reaches
the minimum\footnote{In fact, due to the diophantic nature of $h_{2}$, its minimum is
either at $\left\lfloor \frac{1}{4}(6q-1)\right\rfloor $ or $\left\lceil \frac{1}{4}(6q-1)\right\rceil $.} at $r=\frac{1}{4}(6q-1)$, next starts to increase for $r\geq\left\lceil \frac{1}{4}(6q-1)\right\rceil $.
However, $h,h_{2}$ are defined for $r\leq q-1$. Thus, it is clear
that within the interval $0\leq r\leq q-1$ the function $h_{2}$
is strictly decreasing with respect to $r$. Moreover, it is easy
to verify that $q-1<\left\lfloor \frac{1}{4}(6q-1)\right\rfloor $.
Thus, to determine the minimal value of $h_{2}$ it is enough to check
their value for $r=q-1$.

Thus $h_{2}$: 
\begin{equation}
h_{2}(q,q-1)=\frac{1}{6}\left(2q^{2}+3q+1\right)\label{eq:lemma:hfun_bound:32}
\end{equation}

Since, $q\geq2$ then it is easy to verify that $h_{2}(q,q-1)>0$.
This implies that $h(q,k,r)>0$ for every $0\leq r\leq q-1$ and $k$
such that $\frac{1}{2}\left(2q^{2}-r\left(2q+1\right)\right)\leq k<\frac{1}{2}\left(2q^{2}-(r-1)\left(2q+1\right)\right)$.
Hence, also $\mathcal{F}(n,\mathcal{X}(n)-k)>0$ for $n=2q+1$ where
$1\leq k\leq q^{2}$,  which completes the proof of (\ref{eq:lemma:hfun_bound:2}). 

\smallskip\textsc{Proof. Equation (\ref{eq:lemma:hfun_bound:3}), part 1.}\smallskip

Let $n$ be even i.e. $n=2q$ where $q\in\mathbb{N}_{+}$. Since (\ref{corollary:no-of-edges-dt-max-graph})
to prove that $\mathcal{F}(n,\mathcal{X}(n)+k)$ is smaller than $0$
it is enough to show that for every integer $q,k$ such that $q\geq2$
and $1\leq k\leq\binom{n}{2}-\mathcal{X}(n)$ where $\binom{n}{2}-\mathcal{X}(n)=\binom{2q}{2}-q(q-1)=q^{2}$
it holds that $\mathcal{F}(2q,q(q-1)+k)\leq0$. After a series of
elementary transformations applied to (\ref{eq:empty-triad-theorem-eq-2})
we obtain that:
\begin{equation}
\mathcal{F}(2q,q(q-1)+k)=-\frac{2}{3}\left(q\left\lfloor \frac{k}{q}\right\rfloor ^{2}+(q-2k)\left\lfloor \frac{k}{q}\right\rfloor +k(q-1)\right)\label{eq:lemma:hfun_bound:33}
\end{equation}

Let us consider the relationship between $k$ and $\left\lfloor \frac{k}{q}\right\rfloor $.
When $1\leq k<q$ it holds that $\left\lfloor \frac{k}{q}\right\rfloor =0$,
when $q\leq k<2q$ it holds that $\left\lfloor \frac{k}{q}\right\rfloor =1$
and similarly, $2q\leq k<3q$ then it holds that $\left\lfloor \frac{k}{q}\right\rfloor =2$.
In general, when $rq\leq k<(r+1)q$ then $\left\lfloor \frac{k}{q}\right\rfloor =r$.
Of course, since $k\leq q^{2}$ then $r\leq q$. Hence, instead of
considering the function $\mathcal{F}$ at once, we may analyze it
in the intervals in which $\left\lfloor \frac{k}{q}\right\rfloor $
is known and constant. Let us define: 

\begin{equation}
f(q,k,r)\overset{\textit{df}}{=}qr^{2}+(q-2k)r+k(q-1)\label{eq:lemma:hfun_bound:34}
\end{equation}

It is easy to see that $f(q,k,r)=-\frac{3}{2}\cdot\mathcal{F}(2q,q(q-1)+k)$
if $rq\leq k<(r+1)q$ for $r=0,\ldots,q-1$. Hence, instead of analyzing
$\mathcal{F}$ we will focus on the auxiliary function $f.$ 

The first observation is that $f$ is linear with respect to $k$
providing that $q$ and $r$ are known and fixed. Thus, the minimal
value of $f$ with respect to $k$ within the interval $rq\leq k<(r+1)q$
is $\min\{f(q,rq,r),f(q,(r+1)q,r)\}$. In other words, it is enough
to check that $f$ is greater than $0$ at both edges of the interval
for $k$. Let us consider $f$ at the lower bound, i.e. for $k=rq$. 

\begin{equation}
f(q,rq,r)=qr(q-r)\label{eq:lemma:hfun_bound:35}
\end{equation}

It is easy to verify that for every $0<r<q$ and $q\geq2$ the value
$f(q,rq,r)>0$. The function $f(q,rq,r)$ reaches $0$ when $r=0$.
Thus, $f(q,rq,r)\geq0$ for every $r$ such that $0\leq r\leq q$. 

Let us consider $f$ at the other end of interval, i.e. for $k=(r+1)q-1$.
\begin{equation}
f(q,(r+1)q-1,r)=q^{2}(r+1)-q\left(r^{2}+2r+2\right)+2r+1\label{eq:lemma:hfun_bound:36}
\end{equation}

Similarly as above, we would like to show that for every admissible
$r$ the function $f(q,(r+1)q-1,r)\geq0$. Hence, let us rewrite $f$
with respect to $r$. 

\begin{equation}
f(q,(r+1)q-1,r)=-qr^{2}+r\left(q^{2}-2q+2\right)+\left(q^{2}-2q+1\right)\label{eq:lemma:hfun_bound:37}
\end{equation}

When considering $f$ as a polynomial with respect to $r$ one may
notice that the coefficient at $r^{2}$ is negative ($-q<0$) which
means that $f$ is concave.

Let us denote $f_{2}(q,r)\stackrel{\textit{df}}{=}f(q,(r+1)q-1,r)$.
It is easy to compute that $\frac{\partial f_{2}}{\partial r}=0$
when $r=\frac{q^{2}-2q+2}{2q}$. Since $\frac{\partial^{2}f_{2}}{\partial r^{2}}=-2q>0$,
thus $f_{2}$ reaches the maximum\footnote{In fact, due to the diophantine nature of $f$ it reaches the maximum
for $r=\left\lfloor \frac{q^{2}-2q+2}{2q}\right\rfloor $ or $r=\left\lceil \frac{q^{2}-2q+2}{2q}\right\rceil $.} for $r=\frac{q^{2}-2q+2}{2q}$. Since the interval of $r$ is $0\leq r<q$
and also $0\leq\frac{q^{2}-2q+2}{2q}<q$ therefore the minimum of
$f_{2}$ for $0\leq r<q$ is the smaller of the two $f_{2}(q,0)$
and $f_{2}(q,q-1)$. 

Hence 
\begin{equation}
f_{2}(q,0)=q^{2}-2q+1,\,\,\,\,f_{2}(q,q-1)=q-1\label{eq:lemma:hfun_bound:38}
\end{equation}

Since for every $q\geq2$ it holds that $\min\{f_{2}(q,0),f_{2}(q,q-1)\}\geq0$
then $f_{2}(q,r)\geq0$ for every fixed $q\geq2$ and $0\leq r<q$,
which implies that also for $k=(r+1)q-1$, $f(q,k,r)\geq0$. Therefore
$f(q,k,r)\geq0$ for every $rq\leq k<(r+1)q$ for $r=0,\ldots,q$. 

As $f(q,k,r)=-\frac{3}{2}\cdot\mathcal{F}(2q,q(q-1)+k)$ when $rq\leq k<(r+1)q$,
then due to the arbitrary choice of $r$ it holds that $\mathcal{F}(n,\mathcal{X}(n)+k)\leq0$
for $n=2q$ and $0\leq k<q^{2}$. As one may observe, the above reasoning
does not cover $k=q^{2}$. This is the last ``point interval'' that
needs to be considered. For $k=q^{2}$ we have 
\begin{equation}
\mathcal{F}(2q,q(q-1)+q^{2})=\frac{1}{3}(-2)q\left(\lfloor2q\rfloor^{2}+(1-4q)\lfloor2q\rfloor+2(2q-1)q\right)\label{eq:lemma:hfun_bound:39}
\end{equation}

Since $q\in\mathbb{N}_{+}$ then $\lfloor2q\rfloor=2q$. Hence it
is easy to verify that 
\begin{equation}
\mathcal{F}(2q,q(q-1)+q^{2})=0\label{eq:lemma:hfun_bound:40}
\end{equation}

Which completes the first part of the proof of (\ref{eq:lemma:hfun_bound:3}).

\smallskip\textsc{Proof. Equation (\ref{eq:lemma:hfun_bound:3}), part 2.}\smallskip

Let $n$ be odd i.e. $n=2q+1$ where $q\in\mathbb{N}_{+}$. Since
(\ref{corollary:no-of-edges-dt-max-graph}) to prove that $\mathcal{F}(n,\mathcal{X}(n)+k)$
is smaller than $0$ it is enough to show that for every integer $q,k$
such that $q\geq2$ and $1\leq k\leq\binom{2q}{2}-q^{2}-1=q^{2}-q-1$
it holds that $\mathcal{F}(2q+1,q^{2}+k)\leq0$. After a series of
elementary transformations applied to (\ref{eq:empty-triad-theorem-eq-2})
we obtain:

\begin{align}
\mathcal{F}(2q+1,q^{2}+k)= & -\frac{1}{3}\left((2q+1)\left\lfloor \frac{2\left(q^{2}+k\right)}{2q+1}\right\rfloor ^{2}\right.\nonumber \\
 & -\left(4k+4q^{2}-2q-1\right)\left\lfloor \frac{2\left(q^{2}+k\right)}{2q+1}\right\rfloor \label{eq:lemma:hfun_bound:41}\\
 & \left.\begin{array}{c}
\\
\\
\\
\end{array}+(2q-1)(3k+(q-1)q)\right)\nonumber 
\end{align}

Since $1\leq k\leq q^{2}-q-1$ we may estimate the upper and the lower
bound for $\left\lfloor \frac{2\left(q^{2}+k\right)}{2q+1}\right\rfloor $
as
\begin{equation}
q-1\leq\left\lfloor \frac{2q^{2}}{2q+1}\right\rfloor +\left\lfloor \frac{2k}{2q+1}\right\rfloor \leq\left\lfloor \frac{2\left(q^{2}+k\right)}{2q+1}\right\rfloor \label{eq:lemma:hfun_bound:42}
\end{equation}

and

\begin{align}
\left\lfloor \frac{2\left(q^{2}+k\right)}{2q+1}\right\rfloor \leq & \left\lfloor \frac{2\left(q^{2}+q^{2}-q-1\right)}{2q+1}\right\rfloor \leq\left\lfloor \frac{4q^{2}}{2q}-\frac{2q+2}{2q+1}\right\rfloor =\label{eq:lemma:hfun_bound:43}\\
 & \left\lfloor 2q-\frac{2q+2}{2q+1}\right\rfloor =\left\lfloor 2q-2\right\rfloor =2q-2\nonumber 
\end{align}

Let us denote $r\overset{\textit{df}}{=}\left\lfloor \frac{2\left(q^{2}+k\right)}{2q+1}\right\rfloor $.
 Thus, $q-1\leq r\leq2q-2$. Let us consider the relationship between
$k$ and $r$. It holds that $\left\lfloor \frac{2\left(q^{2}+k\right)}{2q+1}\right\rfloor =r$
wherever $r\leq\frac{2\left(q^{2}+k\right)}{2q+1}<r+1$. Thus it is
easy to determine that $\left\lfloor \frac{2\left(q^{2}+k\right)}{2q+1}\right\rfloor =r$
wherever $\frac{1}{2}\left(2qr+r-2q^{2}\right)\leq k<\frac{1}{2}\left(\left(r+1\right)\left(2q+1\right)-2q^{2}\right)$. 

Let us consider the function $\mathcal{F}(2q+1,q^{2}+k)$ for $k\in\mathbb{N}_{+}$
such that $\frac{1}{2}\left(2qr+r-2q^{2}\right)\leq k<\frac{1}{2}\left(\left(r+1\right)\left(2q+1\right)-2q^{2}\right)$.
For this purpose, let us define $f$ 

\begin{equation}
f(q,k,r)\overset{\textit{df}}{=}(2q+1)r^{2}-r\left(4k+4q^{2}-2q-1\right)+(2q-1)(3k+(q-1)q)\label{eq:lemma:hfun_bound:44}
\end{equation}

It is easy to verify that 
\begin{equation}
\mathcal{F}(2q+1,q^{2}+k)=-\frac{1}{3}f(q,k,r)\label{eq:lemma:hfun_bound:45}
\end{equation}

providing that $q,r\in\mathbb{N}_{+}$, $\frac{1}{2}\left(2qr+r-2q^{2}\right)\leq k<\frac{1}{2}\left(\left(r+1\right)\left(2q+1\right)-2q^{2}\right)$,
$q-1\leq r\leq2q-2$ and $q\geq2$. Hence, wherever $f(q,k,r)\geq0$
then $\mathcal{F}(2q+1,q^{2}+k)\leq0$. Let us observe that $f$ is
linear with respect to $k$.  Therefore it is enough to check the
value of $f(q,k,r)$ at the edges of the admissible interval for $k$,
and prove that those values are above $0$ in any possible interval
determined by $r$. For this purpose let us define 
\begin{equation}
f_{2}(q,r)\stackrel{\textit{df}}{=}f(q,\frac{1}{2}\left(2qr+r-2q^{2}\right),r)\label{eq:lemma:hfun_bound:46}
\end{equation}

for the lower bound, and 
\begin{equation}
f_{3}(q,r)\stackrel{\textit{df}}{=}f(q,\frac{1}{2}\left(\left(r+1\right)\left(2q+1\right)-2q^{2}\right)-1,r)\label{eq:lemma:hfun_bound:47}
\end{equation}

for the upper bound. Hence

\begin{equation}
f_{2}(q,r)=-\frac{1}{2}(2q+1)\left(4q^{2}-6qr-2q+2r^{2}+r\right)\label{eq:lemma:hfun_bound:48}
\end{equation}

\begin{equation}
f_{3}(q,r)=-4q^{3}+6q^{2}(r+1)-q\left(2r^{2}+2r+5\right)+\frac{1}{2}\left(-2r^{2}+3r+3\right)\label{eq:lemma:hfun_bound:49}
\end{equation}

Let us reorganize the above equations with respect to $r$:

\begin{equation}
f_{2}(q,r)=-\left(2q+1\right)r^{2}+\left(2q+6q^{2}-\frac{1}{2}\right)r-4q^{3}+q\label{eq:lemma:hfun_bound:50}
\end{equation}

\begin{equation}
f_{3}(q,r)=-\left(2q+1\right)r^{2}+\left(6q^{2}-2q+\frac{3}{2}\right)r-4q^{3}+6q^{2}-5q+\frac{3}{2}\label{eq:lemma:hfun_bound:51}
\end{equation}

Since both $f_{2}$ and $f_{3}$ have second degree polynomials with
respect to $r$, and the coefficients nearby $r^{2}$ are negative,
then $f_{2}$ and $f_{3}$ are concave parabolas. Therefore $f_{2}$
and $f_{3}$ are not smaller than $0$ within the interval $q-1\leq r\leq2q-2$
if they are not negative at both ends of the interval i.e. $q-1$
and $2q-2$. As the estimation (\ref{eq:lemma:hfun_bound:42}) is
not perfect, let us assume for a moment that $r$ is in $q\leq r\leq2q-2$,
whilst the case $r=q-1$ we handle separately. 

Let us examine (\ref{eq:lemma:hfun_bound:50}). 

\begin{equation}
f_{2}(q,r)=q^{2}+\frac{q}{2}\,\,\,\textit{when}\,\,\,\,r=q\label{eq:lemma:hfun_bound:52}
\end{equation}

and 

\begin{equation}
f_{2}(q,r)=(2q-3)(2q+1)\,\,\,\textit{when}\,\,\,r=2q-2\label{eq:lemma:hfun_bound:53}
\end{equation}

Since $q\geq2$ both of the above equations are greater than $0$.
For (\ref{eq:lemma:hfun_bound:51}) it is enough to assume that $q-1\leq r\leq2q-2$.
Thus,

\begin{equation}
f_{3}(q,r)=q^{2}-\frac{3q}{2}-1\,\,\,\textit{when}\,\,\,r=q-1\label{eq:lemma:hfun_bound:54}
\end{equation}

and 
\begin{equation}
f_{3}(q,r)=2q^{2}+2q-\frac{11}{2}\,\,\,\textit{when}\,\,\,r=2q-2\label{eq:lemma:hfun_bound:55}
\end{equation}

Similarly, it is easy to verify that both of the above expressions
are non negative as $q\geq2$. 

 At the end, let us explicitly calculate 
\begin{equation}
f(q,k,q-1)=2kq+k\label{eq:lemma:hfun_bound:56}
\end{equation}

As $k$ is always non negative, then also in this case $f$ is non
negative $0$. Thereby for every $1\leq k\leq q^{2}-q-1$ it holds
that $\mathcal{F}(2q+1,q^{2}+k)\leq0$ which completes the proof of
the Lemma \ref{lemma:f-function}\QED

\section{\label{sec:appendix-B}Proof of the Lemma \ref{lemma:c-function}}

\noindent\textsc{Thesis.} 

For every $n\in\mathbb{N}_{+},n\geq3$ the function $\mathcal{C}$:
\begin{enumerate}
\item is constant and equals $\mathcal{C}(n,m)=0$ for every $m$ such that
$0\leq m<n$
\item is strictly increasing for every $m\in\mathbb{N}_{+}$ such that $n\leq m\leq\binom{n}{2}$,
i.e. 
\end{enumerate}
\[
\mathcal{C}(n,m+1)-\mathcal{C}(n,m)>0\tag{\ref{eq:lemma:cfun_increasing}}
\]
\noindent\textsc{Proof. Claim 1.}\smallskip 

The first claim that $\mathcal{C}(n,m)=0$ for every $m$ such that
$0\leq m<n$ is a direct consequence of the equation (\ref{eq:number-of-enforced-triads}).
It is enough to note that the right side of expression (\ref{eq:number-of-enforced-triads})
is the product where the first part is $\frac{1}{2}\left\lfloor \frac{m}{n}\right\rfloor $.
Hence, wherever $m<n$ the product often equals $0$. 

\noindent\textsc{Proof. Claim 2.}\smallskip 

Due to (Theorem \ref{theorem:enforced-triads_theorem}) it holds that
\begin{align}
\mathcal{C}(n,m+1)-\mathcal{C}(n,m)= & \frac{1}{2}\left(\left\lfloor \frac{m}{n}\right\rfloor \left(n\left\lfloor \frac{m}{n}\right\rfloor -2m+n\right)-\right.\nonumber \\
 & \left.\left\lfloor \frac{m+1}{n}\right\rfloor \left(n\left\lfloor \frac{m+1}{n}\right\rfloor -2m+n-2\right)\right)\label{eq:lemma-8-eq-1}
\end{align}

It is easy to observe that for some positive integer $p=1,2,\ldots$
when $m=np-1$ then $\left\lfloor \frac{m}{n}\right\rfloor =p-1,\left\lfloor \frac{m+1}{n}\right\rfloor =p$.
Next, by increasing $m$ by one we get $m=np$ and $\left\lfloor \frac{m}{n}\right\rfloor =p,\left\lfloor \frac{m+1}{n}\right\rfloor =p$.
Then, for $m=n(p+1)-1$ the values of our floored expressions change
to $\left\lfloor \frac{m}{n}\right\rfloor =p,\left\lfloor \frac{m+1}{n}\right\rfloor =p+1$,
and then by increasing $m$ by one we get $\left\lfloor \frac{m}{n}\right\rfloor =p+1,\left\lfloor \frac{m+1}{n}\right\rfloor =p+1$.
Hence, there are two different intervals with respect to the values
$\left\lfloor \frac{m}{n}\right\rfloor $ and $\left\lfloor \frac{m+1}{n}\right\rfloor $.
The first one in which both expressions have the same value, and the
other one (composed of one point) in which their values differ by
one. In general, we may observe that: 

wherever $m=np-1$ then $\left\lfloor \frac{m}{n}\right\rfloor =p-1,\left\lfloor \frac{m+1}{n}\right\rfloor =p$,
and wherever $np\leq m<n(p+1)-1$ then $\left\lfloor \frac{m}{n}\right\rfloor =p,\left\lfloor \frac{m+1}{n}\right\rfloor =p$.

Let us define the auxiliary function $h$ by replacing in (\ref{eq:lemma-8-eq-1})
$\left\lfloor \frac{m}{n}\right\rfloor $ by $r$ and $\left\lfloor \frac{m+1}{n}\right\rfloor $
by $t$: 

\begin{equation}
h(n,m,r,t)\overset{\textit{df}}{=}\frac{1}{2}\left(r\left(nr-2m+n\right)-t\left(nt-2m+n-2\right)\right)\label{eq:lemma-8-eq-2}
\end{equation}

The function $h$ can be rewritten with respect to $m$, so

\begin{equation}
h(n,m,r,t)=\frac{1}{2}nr^{2}+m\left(t-r\right)+\frac{1}{2}nr-\frac{1}{2}nt^{2}-\frac{1}{2}nt+t\label{eq:lemma-8-eq-3}
\end{equation}

It is easy to observe that 
\begin{equation}
\mathcal{C}(n,m+1)-\mathcal{C}(n,m)=h(n,m,r,t)\label{eq:lemma-8-eq-4}
\end{equation}

where $r=\left\lfloor \frac{m}{n}\right\rfloor $ and $t=\left\lfloor \frac{m+1}{n}\right\rfloor $.
Thus, instead of analyzing $h(n,m,r,t)$ for $m$ such that $n\leq m\leq\binom{n}{2}$
we analyze $h(n,m,r,t)$ in two intervals $m=np-1$ and $np\leq m<n(p+1)-1$.
This, due to the arbitrary choice of $p$, would apply to $\mathcal{C}(n,m+1)-\mathcal{C}(n,m)$
over the whole interval $n\leq m\leq\binom{n}{2}$. 

Let us observe that $h$ is linear with respect to $m$. Thus to prove
that $h(n,m,r,t)>0$ when $n,r,t$ are constant, one needs only to
verify the value of $h$ at the ends of both intervals to which $m$
may belong. Thus, let us consider the first ``point'' interval $m=np-1$.
In this interval $\left\lfloor \frac{m}{n}\right\rfloor =p-1,\left\lfloor \frac{m+1}{n}\right\rfloor =p$,
thus: 

\begin{equation}
h(n,np-1,p-1,p)=p-1\label{eq:lemma-8-eq-5}
\end{equation}

As $m\geq n$, and $m=np-1$ thus $p\geq2$. Hence, 
\begin{equation}
h(n,np-1,p-1,p)\geq2-1=1\label{eq:lemma-8-eq-6}
\end{equation}

This supports the thesis of the theorem, i.e. $np\leq m<n(p+1)-1$,
where $\left\lfloor \frac{m}{n}\right\rfloor =p,\left\lfloor \frac{m+1}{n}\right\rfloor =p$.
For both its ends we have:
\begin{equation}
h(n,np,p,p)=p\label{eq:lemma-8-eq-7}
\end{equation}

\begin{equation}
h(n,n(p+1)-1,p,p)=p\label{eq:lemma-8-eq-8}
\end{equation}

As $m\geq n$ and $np\leq m$ then $p\geq1$. Thus in both cases $h$
is strictly greater than $0$. Hence, for every $np-1\leq m\leq n(p+1)-1$
it holds that 
\begin{equation}
\mathcal{C}(n,m+1)-\mathcal{C}(n,m)>0\label{eq:lemma-8-eq-9}
\end{equation}

Due to the arbitrary choice of $p$ this statement completes the proof
of the theorem. \QED

\section{\label{sec:appendix:C}Proof of the Lemma \ref{lemma:c-function-equality}}

\noindent\textsc{Thesis.}\smallskip

For every $n\in\mathbb{N}_{+},n\geq3$ it holds that 
\[
\binom{n}{3}-\mathcal{C}(n,\mathcal{X}(n))=\mathcal{Y}(n)\tag{\ref{eq:lemma:cfun_equality}}
\]
\smallskip\textsc{Proof. Part 1.}\smallskip

Let $n=4q$ ($n$ is even, and $\left\lfloor \frac{n}{2}\right\rfloor =\left\lceil \frac{n}{2}\right\rceil =2q$
is even), $n\geq4$, hence $q\geq1$ and $\mathcal{X}(4q)=2q(2q-1)$.
Thus to prove (\ref{eq:lemma:cfun_equality}) for even numbers we
show that 
\begin{equation}
\binom{4q}{3}-\mathcal{C}(4q,2q(2q-1))-\mathcal{Y}(4q)=0\label{eq:lemma-C-eq-1}
\end{equation}

Since (\ref{eq:yfun-def}) reduces to: 
\begin{align}
\mathcal{Y}(4q)= & \binom{4q}{3}-\left(\binom{2q}{3}-\frac{q\left(q^{2}-1\right)}{3}\right)\nonumber \\
 & -\left(\binom{2q}{3}-\frac{q\left(q^{2}-1\right)}{3}\right)\label{eq:lemma-C-eq-2}
\end{align}

by elementary transformations one may show that (\ref{eq:lemma-C-eq-1})
is equivalent to 
\begin{equation}
2q\left(\left\lceil \frac{1}{2}-q\right\rceil +q-1\right)^{2}=0\label{eq:lemma-C-eq-3}
\end{equation}

The above is true as $\left\lceil \frac{1}{2}-q\right\rceil =1-q$
for every $q\in\mathbb{N}_{+}$. 

\smallskip\textsc{Proof. Part 2.}\smallskip

Let $n=4q+1$ ($n$ is odd, $\left\lfloor \frac{n}{2}\right\rfloor =2q$
is even, and $\left\lceil \frac{n}{2}\right\rceil =2q+1$ is odd),
$n\geq4$, hence $q\geq1$ and $\mathcal{X}(4q+1)=\binom{\left\lfloor \frac{n}{2}\right\rfloor }{2}+\binom{\left\lceil \frac{n}{2}\right\rceil }{2}=\binom{2q}{2}+\binom{2q+1}{2}=4q^{2}$.
Thus to prove (\ref{eq:lemma:cfun_equality}) for $n=4q+1$ we show
that 
\begin{equation}
\binom{4q+1}{3}-\mathcal{C}(4q+1,4q^{2})-\mathcal{Y}(4q+1)=0\label{eq:lemma-C-eq-4}
\end{equation}

Since (\ref{eq:yfun-def}) reduces to: 

\begin{align}
\mathcal{Y}(4q+1)= & \binom{4q+1}{3}-\left(\binom{2q}{3}-\frac{q\left(q^{2}-1\right)}{3}\right)\nonumber \\
 & -\left(\binom{2q+1}{3}-\frac{q\left(2q^{2}+3q+1\right)}{6}\right)\label{eq:lemma-C-eq-5}
\end{align}

by elementary transformations one may show that (\ref{eq:lemma-C-eq-4})
is equivalent to 

\begin{align}
\frac{1}{2}\left((4q+1)\left\lfloor \frac{4q^{2}}{4q+1}\right\rfloor ^{2}+\right.\nonumber \\
\left.\left(-8q^{2}+4q+1\right)\left\lfloor \frac{4q^{2}}{4q+1}\right\rfloor +q\left(4q^{2}-5q+1\right)\right) & =0\label{eq:lemma-C-eq-6}
\end{align}

Let us note that for every $q\geq1$ it holds\footnote{compare with (\ref{eq:lemma:hfun_bound:13}).}
that $\left\lfloor \frac{4q^{2}}{4q+1}\right\rfloor =q-1$. Thus,
the above equation can be written in the form 
\begin{equation}
\frac{1}{2}\left(\left(-8q^{2}+4q+1\right)(q-1)+\left(4q^{2}-5q+1\right)q+(4q+1)(q-1)^{2}\right)=0\label{eq:lemma-C-eq-7}
\end{equation}

which can be easily verified as true. 

\smallskip\textsc{Proof. Part 3.}\smallskip

Let $n=4q+2$ ($n$ is even, $\left\lfloor \frac{n}{2}\right\rfloor =2q+1$
is odd, and $\left\lceil \frac{n}{2}\right\rceil =2q+1$ is odd) and
$\mathcal{X}(4q+2)=\binom{\left\lfloor \frac{n}{2}\right\rfloor }{2}+\binom{\left\lceil \frac{n}{2}\right\rceil }{2}=\binom{2q+1}{2}+\binom{2q+1}{2}=2q(2q+1)$
Thus, to prove (\ref{eq:lemma:cfun_equality}) for $n=4q+2$ we show
that 
\begin{equation}
\binom{4q+2}{3}-\mathcal{C}(4q+2,2q(2q+1))-\mathcal{Y}(4q+2)=0\label{eq:lemma-C-eq-8}
\end{equation}

Since (\ref{eq:yfun-def}) reduces to: 
\begin{equation}
\mathcal{Y}(4q+2)=\binom{4q+2}{3}-2\left(\binom{2q+1}{3}-\frac{q\left(2q^{2}+3q+1\right)}{6}\right)\label{eq:lemma-C-eq-9}
\end{equation}

by elementary transformations one may show that (\ref{eq:lemma-C-eq-8})
is equivalent to 
\begin{equation}
(2q+1)\left(\lfloor q\rfloor^{2}+(1-2q)\lfloor q\rfloor+(q-1)q\right)=0\label{eq:lemma-C-eq-10}
\end{equation}

As $q$ is an integer it is easy to show that (\ref{eq:lemma-C-eq-10})
is true. 

\smallskip\textsc{Proof. Part 4.}\smallskip

Let $n=4q+3$ ($n$ is odd $\left\lfloor \frac{n}{2}\right\rfloor =2q+1$
is odd, and $\left\lceil \frac{n}{2}\right\rceil =2q+2$ is even)
and $\mathcal{X}(4q+3)=\binom{\left\lfloor \frac{n}{2}\right\rfloor }{2}+\binom{\left\lceil \frac{n}{2}\right\rceil }{2}=\binom{2q+1}{2}+\binom{2q+2}{2}=(2q+1)^{2}$.
Thus, to prove (\ref{eq:lemma:cfun_equality}) for $n=4q+3$ we show
that 

\begin{equation}
\binom{4q+3}{3}-\mathcal{C}(4q+3,(2q+1)^{2})-\mathcal{Y}(4q+3)=0\label{eq:lemma-C-eq-11}
\end{equation}

by elementary transformations one may show that (\ref{eq:lemma-C-eq-11})
is equivalent to:
\begin{align}
\frac{1}{2}\left(\left(-8q^{2}-4q+1\right)\left\lfloor \frac{(2q+1)^{2}}{4q+3}\right\rfloor \right.+\nonumber \\
\left.(4q+3)\left\lfloor \frac{(2q+1)^{2}}{4q+3}\right\rfloor ^{2}+\left(4q^{2}+q-1\right)q\right) & =0\label{eq:lemma-C-eq-12}
\end{align}

Since\footnote{Let us notice that $\left\lfloor \frac{(2q+1)^{2}}{4q+3}\right\rfloor =\left\lfloor \frac{4q^{2}+4q+1}{4q+3}\right\rfloor =\ldots=\left\lfloor q+\frac{q+1}{4q+3}\right\rfloor $.
The fact that for $q=0,1,\ldots$ the expression $\frac{q+1}{4q+3}$
is always smaller than $1$, implies that $\left\lfloor \frac{(2q+1)^{2}}{4q+3}\right\rfloor =\left\lfloor q\right\rfloor $. } $\left\lfloor \frac{(2q+1)^{2}}{4q+3}\right\rfloor =\left\lfloor q\right\rfloor =q$
then  the above expression can be written as: 
\begin{equation}
\frac{1}{2}\left((4q+3)q^{2}+\left(4q^{2}+q-1\right)q+\left(-8q^{2}-4q+1\right)q\right)=0\label{eq:lemma-C-eq-13}
\end{equation}

which can easily be verified as true. This also completes the proof
of the Lemma \ref{lemma:c-function-equality}.

\QED

\section{\label{sec:appendix:D}Proof of the Lemma \ref{lemma:g-function}}

\noindent\textsc{Thesis.}\smallskip

For every $n\in\mathbb{N}_{+},n\geq3$ the function $\mathcal{G}$
is strictly decreasing for every $m\in\mathbb{N}_{+}$ such that $1\leq m\leq\mathcal{X}(n)$,
i.e. 
\[
\mathcal{G}(n,m)-\mathcal{G}(n,m+1)>0\,\,\,\textit{where}\,\,\,1\leq m<\mathcal{X}(n)\tag{\ref{eq:lemma:cfun_decreasing}}
\]
\noindent\textsc{Proof of (\ref{eq:lemma:cfun_decreasing}), part 1 (for even numbers)}\smallskip

Let $n=2q$ (even), $n\geq3$, hence $q\geq2$, and $m,m+1\leq\mathcal{X}(2q)=q(q-1)$.
Note that, in particular, the last assumption implies that $m\leq q(q-1)-1$.
Hence (\ref{eq:lemma:cfun_decreasing}) can be written as:

\begin{align}
3\left(\mathcal{G}(n,m)-\mathcal{G}(n,m+1)\right)= & -2q\left\lfloor \frac{m}{q}\right\rfloor ^{2}+(4m-2q)\left\lfloor \frac{m}{q}\right\rfloor +2q\left\lfloor \frac{m+1}{q}\right\rfloor ^{2}\nonumber \\
 & -3q\left\lfloor \frac{m}{2q}\right\rfloor ^{2}+3q\left\lfloor \frac{m+1}{2q}\right\rfloor ^{2}-4m\left\lfloor \frac{m+1}{q}\right\rfloor \nonumber \\
 & +2q\left\lfloor \frac{m+1}{q}\right\rfloor -4\left\lfloor \frac{m+1}{q}\right\rfloor +3(m-q)\left\lfloor \frac{m}{2q}\right\rfloor \nonumber \\
 & -3m\left\lfloor \frac{m+1}{2q}\right\rfloor +3q\left\lfloor \frac{m+1}{2q}\right\rfloor -3\left\lfloor \frac{m+1}{2q}\right\rfloor +6q-6\label{eq:lemma-d:eq:1}
\end{align}

Let us denote $r_{1}=\left\lfloor \frac{m}{q}\right\rfloor ,r_{2}=\left\lfloor \frac{m}{2q}\right\rfloor ,r_{3}=\left\lfloor \frac{m+1}{q}\right\rfloor ,r_{4}=\left\lfloor \frac{m+1}{2q}\right\rfloor $.
This allows us to denote 
\begin{align}
3\left(\mathcal{G}(n,m)-\mathcal{G}(n,m+1)\right)= & -2qr_{1}^{2}+(4m-2q)r_{1}+2qr_{3}^{2}-3qr_{2}^{2}\nonumber \\
 & +3qr_{4}^{2}-4mr_{3}+2qr_{3}-4r_{3}+3(m-q)r_{2}\nonumber \\
 & -3mr_{4}+3qr_{4}-3r_{4}+6q-6\label{eq:lemma-d:eq:2}
\end{align}

Let us introduce the auxiliary function $h$ such that 

\begin{align}
h(q,m,r_{1},r_{2},r_{3},r_{4})\overset{\textit{df}}{=} & r_{1}(4m-2q)+3r_{2}(m-q)-4mr_{3}-3mr_{4}\nonumber \\
 & -2qr_{1}^{2}-3qr_{2}^{2}+2qr_{3}^{2}+3qr_{4}^{2}+2qr_{3}\nonumber \\
 & +3qr_{4}+6q-4r_{3}-3r_{4}-6\label{eq:lemma-d:eq:3}
\end{align}

It is easy to verify that 
\begin{equation}
3\left(\mathcal{G}(n,m)-\mathcal{G}(n,m+1)\right)=h(q,m,r_{1},r_{2},r_{3},r_{4})\label{eq:lemma-d:eq:4}
\end{equation}

Let us try to investigate changes in the values $r_{1},r_{2},r_{3}$
and $r_{4}$. To do so, let us create the following table: 

\medskip{}

\begin{center}
\begin{tabular}{|c|c|c|c|c|}
\hline 
interval of $m$ & $\left\lfloor \frac{m}{q}\right\rfloor $ & $\left\lfloor \frac{m}{2q}\right\rfloor $ & $\left\lfloor \frac{m+1}{q}\right\rfloor $ & $\left\lfloor \frac{m+1}{2q}\right\rfloor $\tabularnewline
\hline 
\hline 
$0q\leq m<1q-1$ & $0$ & $0$ & $0$ & $0$\tabularnewline
\hline 
$1q-1=m$ & $0$ & $0$ & $1$ & $0$\tabularnewline
\hline 
$1q\leq m<2q-1$ & $1$ & $0$ & $1$ & $0$\tabularnewline
\hline 
$2q-1=m$ & $1$ & $0$ & $2$ & $1$\tabularnewline
\hline 
$2q\leq m<3q-1$ & $2$ & $1$ & $2$ & $1$\tabularnewline
\hline 
$3q-1=m$ & $2$ & $1$ & $3$ & $1$\tabularnewline
\hline 
$3q\leq m<4q-1$ & $3$ & $1$ & $3$ & $1$\tabularnewline
\hline 
$4q-1=m$ & $3$ & $1$ & $4$ & $2$\tabularnewline
\hline 
$4q\leq m<5q-1$ & $4$ & $2$ & $4$ & $2$\tabularnewline
\hline 
\end{tabular}
\par\end{center}

\medskip{}

As we can see, there are four kinds of interval (hereinafter referred
to as cases) that need to be considered with respect to $m$. Every
analyzed interval is parametrized by the auxiliary variable $s\in\mathbb{N}\cup\{0\}$.
By choosing arbitrarily $s=0,1,2,3,\ldots$ we are able to analyze
the function $h$, and as follows $\mathcal{G}(n,m)-\mathcal{G}(n,m+1)$,
for every interesting $m$. The cases we need to consider are: 

\medskip{}

\begin{center}
\begin{tabular}{|c|c|c|c|c|c|}
\hline 
Case & interval of $m$ & $\left\lfloor \frac{m}{q}\right\rfloor $ & $\left\lfloor \frac{m}{2q}\right\rfloor $ & $\left\lfloor \frac{m+1}{q}\right\rfloor $ & $\left\lfloor \frac{m+1}{2q}\right\rfloor $\tabularnewline
\hline 
\hline 
1a & $2sq\leq m<(2s+1)q-1$ & $2s$ & $s$ & $2s$ & $s$\tabularnewline
\hline 
2a & $(2s+1)q-1=m$ & $2s$ & $s$ & $2s+1$ & $s+1$\tabularnewline
\hline 
3a & $\left(2s+1\right)q\leq m<(2s+2)q-1$ & $2s+1$ & $s$ & $2s+1$ & $s$\tabularnewline
\hline 
4a & $(2s+1)q-1=m$ & $2s$ & $s$ & $2s+1$ & $s$\tabularnewline
\hline 
\end{tabular}
\par\end{center}

\smallskip\textsc{Case 1a}\smallskip

Let $2sq\leq m<(2s+1)q-1$. As $m\leq q(q-1)-1$, then the candidate
for the highest value of $s$ is the smallest integer for which $q(q-1)-1<(2s+1)q-1$,
hence $\frac{q-2}{2}<s$. This means that $\left\lfloor \frac{q-2}{2}\right\rfloor +1=s$,
hence $\frac{q-2}{2}+1\geq s$. On the other hand, as $2sq\leq m$
and $m\leq q(q-1)-1$ then $s\leq\frac{q(q-1)-1}{2q}$. Since the
second condition is more restrictive\footnote{Note that $\left(\frac{q-2}{2}+1\right)-\frac{q(q-1)-1}{2q}=\frac{1+q}{2q}$ }
we assume that $s\leq\frac{q(q-1)-1}{2q}$.  Let us denote 
\begin{equation}
h(q,m,r_{1},r_{2},r_{3},r_{4})=h(q,m,2s,s,2s,s)\label{eq:lemma-d:eq:5}
\end{equation}

Hence,

\begin{equation}
h(q,m,2s,s,2s,s)=6q-11s-6\label{eq:lemma-d:eq:6}
\end{equation}

The highest possible value of $s$ is $\frac{q(q-1)-1}{2q}$, hence
the minimal value of $h$ providing this constraint is $6(q-1)-11\frac{q(q-1)-1}{2q}$
i.e.

\begin{equation}
h(q,m,2s,s,2s,s)\geq6(q-1)-11\frac{q(q-1)-1}{2q}\label{eq:lemma-d:eq:7}
\end{equation}

Which is equivalent to 

\begin{equation}
h(q,m,2s,s,2s,s)\geq\frac{q^{2}-q+11}{2q}\label{eq:lemma-d:eq:8}
\end{equation}

Hence, it is clear that for $q\geq2$ the right side of the above
equation is always greater than $0$. 

\smallskip\textsc{Case 2a}\smallskip

Let $(2s+1)q-1=m$. Since $m\leq q(q-1)-1$ then $s$ cannot be higher
than the maximal integer which meets the inequality $(2s+1)q-1\leq q(q-1)-1$,
i.e. $s\leq\frac{q-2}{2}$.  Let us calculate $h$, for $m=(2s+1)q-1$,
$r_{1}=2s,r_{2}=s,r_{3}=2s+1$ and $r_{4}=s+1$. 

\begin{equation}
h(q,m,r_{1},r_{2},r_{3},r_{4})=9q-11s-6\label{eq:lemma-d:eq:9}
\end{equation}

As the maximal $s=\frac{q-2}{2}$ then 

\begin{equation}
h(q,m,r_{1},r_{2},r_{3},r_{4})\geq9q-11\frac{q-2}{2}-6\label{eq:lemma-d:eq:10}
\end{equation}

which is equivalent to 
\begin{equation}
h(q,m,r_{1},r_{2},r_{3},r_{4})\geq\frac{7q}{2}+5\label{eq:lemma-d:eq:11}
\end{equation}

It is clear that for $q\geq2$ the right side of the above equation
is always greater than $0$. 

\smallskip\textsc{Case 3a}\smallskip

Let $\left(2s+1\right)q\leq m<(2s+2)q-1$

Since $m\leq q(q-1)-1$ then $s$ is not higher than the maximal integer
which meets the inequality $q(q-1)-1<(2s+2)q-1$, i.e. $\frac{q-3}{2}<s$.
 Thus, $s=\left\lfloor \frac{q-3}{2}\right\rfloor +1$, hence $s\leq\frac{q-3}{2}+1$.
On the other hand, also $\left(2s+1\right)q\leq m$ and $m\leq q(q-1)-1$.
Thus $s$ should meet $\left(2s+1\right)q\leq q(q-1)-1$, i.e. $s\leq\frac{1}{2}\left(\frac{q(q-1)-1}{q}-1\right)$.
The second condition is more restrictive\footnote{as $\left(\frac{q-3}{2}+1\right)-\frac{1}{2}\left(\frac{q(q-1)-1}{q}-1\right)=\frac{q+1}{2q}$}
hence we assume that $s\leq\frac{1}{2}\left(\frac{q(q-1)-1}{q}-1\right)$.
 Let us calculate $h$ assuming $r_{1}=2s+1,r_{2}=s,r_{3}=2s+1$,
and $r_{4}=s$. So, 
\begin{equation}
h(q,m,r_{1},r_{2},r_{3},r_{4})=h(q,m,2s+1,s,2s+1,s)\label{eq:lemma-d:eq:12}
\end{equation}

and thus,

\begin{equation}
h(q,m,2s+1,s,2s+1,s)=6q-11s-10\label{eq:lemma-d:eq:13}
\end{equation}

The highest allowed value of $s$ is $\frac{1}{2}\left(\frac{q(q-1)-1}{q}-1\right)$,
thus it is true that 

\begin{equation}
h(q,m,2s+1,s,2s+1,s)\geq6q-\frac{11}{2}\left(\frac{q(q-1)-1}{q}-1\right)-10\label{eq:lemma-d:eq:14}
\end{equation}

which is equivalent to 

\begin{equation}
h(q,m,2s+1,s,2s+1,s)\geq\frac{1}{2}\left(q+\frac{11}{q}+2\right)\label{eq:lemma-d:eq:15}
\end{equation}

It is clear that for $q\geq2$ the above equation is always greater
than $0$.

\smallskip\textsc{Case 4a}\smallskip

Let $(2s+1)q-1=m$

Since $m\leq q(q-1)-1$ then $s$ cannot be higher than the maximal
integer which meets the inequality $(2s+1)q-1\leq q(q-1)-1$, i.e.
$s\leq\frac{q-2}{2}$.  Let us calculate $h$, by the assumptions
that $m=(2s+1)q-1$, $r_{1}=2s,r_{2}=s,r_{3}=2s+1$ and $r_{4}=s$. 

\begin{equation}
h(q,m,r_{1},r_{2},r_{3},r_{4})=6q-11s-6\label{eq:lemma-d:eq:16}
\end{equation}

Since the maximal $s$ is $\frac{q-2}{2}$ then 

\begin{equation}
h(q,m,r_{1},r_{2},r_{3},r_{4})\geq6q-11\left(\frac{q-2}{2}\right)-6\label{eq:lemma-d:eq:17}
\end{equation}

which is equivalent to 

\begin{equation}
h(q,m,r_{1},r_{2},r_{3},r_{4})\geq\frac{q}{2}+5\label{eq:lemma-d:eq:18}
\end{equation}

It is clear that for $q\geq2$ the above equation is always greater
than $0$. This remark completes the proof for $n=2q$. 

\smallskip\textsc{Proof of (\ref{eq:lemma:cfun_decreasing}), part 2 (for odd numbers)}\smallskip

Let $n=2q+1$ (odd), $n\geq3$, hence $q\geq1$ and $0\leq m,m+1\leq\mathcal{X}(2q+1)=q^{2}$.
In particular, the last assumption implies that $0\leq m\leq q^{2}-1$.
When $n=2q+1$ it holds that:
\begin{align}
6\left(\mathcal{G}(n,m)-\mathcal{G}(n,m+1)\right)= & -6+12q+(6m-6q-3)\left\lfloor \frac{m}{2q+1}\right\rfloor \nonumber \\
 & -3(2q+1)\left\lfloor \frac{m}{2q+1}\right\rfloor ^{2}+(8m-4q-2)\left\lfloor \frac{2m}{2q+1}\right\rfloor \nonumber \\
 & -2(2q+1)\left\lfloor \frac{2m}{2q+1}\right\rfloor ^{2}-3\left\lfloor \frac{m+1}{2q+1}\right\rfloor \nonumber \\
 & -6m\left\lfloor \frac{m+1}{2q+1}\right\rfloor +6q\left\lfloor \frac{m+1}{2q+1}\right\rfloor +3\left\lfloor \frac{m+1}{2q+1}\right\rfloor ^{2}\nonumber \\
 & +6q\left\lfloor \frac{m+1}{2q+1}\right\rfloor ^{2}-6\left\lfloor \frac{2(m+1)}{2q+1}\right\rfloor -4q\left\lfloor \frac{2(m+1)}{2q+1}\right\rfloor ^{2}\nonumber \\
 & +2\left\lfloor \frac{2(m+1)}{2q+1}\right\rfloor ^{2}+8m\left\lfloor \frac{2(m+1)}{2q+1}\right\rfloor +4q\left\lfloor \frac{2(m+1)}{2q+1}\right\rfloor \label{eq:lemma-d:eq:19}
\end{align}

Let us denote $r_{1}=\left\lfloor \frac{2m}{2q+1}\right\rfloor $,
$r_{2}=\left\lfloor \frac{m}{2q+1}\right\rfloor $, $r_{3}=\left\lfloor \frac{2(m+1)}{2q+1}\right\rfloor $
and $r_{4}=\left\lfloor \frac{m+1}{2q+1}\right\rfloor $. This allows
us to simplify the above equation to 
\begin{align}
6\left(\mathcal{G}(n,m)-\mathcal{G}(n,m+1)\right)= & -6+12q+r_{1}(8m-4q-2)\nonumber \\
 & -2(2q+1)r_{1}^{2}+r_{2}(6m-6q-3)\nonumber \\
 & -8mr_{3}-6mr_{4}+3(2q+1)r_{2}^{2}+4qr_{3}^{2}+6qr_{4}^{2}\nonumber \\
 & +4qr_{3}+6qr_{4}+2r_{3}^{2}+3r_{4}^{2}-6r_{3}-3r_{4}\label{eq:lemma-d:eq:20}
\end{align}

Let us define:

\begin{align}
h(q,m,r_{1},r_{2},r_{3},r_{4})= & -6+12q+r_{1}(8m-4q-2)-2(2q+1)r_{1}^{2}\nonumber \\
 & +r_{2}(6m-6q-3)-8mr_{3}-6mr_{4}\nonumber \\
 & +3(2q+1)r_{2}^{2}+4qr_{3}^{2}+6qr_{4}^{2}\nonumber \\
 & +4qr_{3}+6qr_{4}+2r_{3}^{2}+3r_{4}^{2}-6r_{3}-3r_{4}\label{eq:lemma-d:eq:21}
\end{align}

It is clear that 
\begin{equation}
6\left(\mathcal{G}(n,m)-\mathcal{G}(n,m+1)\right)>0\Leftrightarrow h(q,m,r_{1},r_{2},r_{3},r_{4})>0\label{eq:lemma-d:eq:22}
\end{equation}

Let us try to investigate changes in the values $r_{1},r_{2},r_{3}$
and $r_{4}$. To do so, let us write down a few cases of each in the
form of a table: 

\medskip{}

\begin{center}
\begin{tabular}{|c|c|}
\hline 
interval & $\left\lfloor \frac{2m}{2q+1}\right\rfloor $\tabularnewline
\hline 
\hline 
$0\leq m<\text{\ensuremath{\frac{1}{2}}}(2q+1)$ & $0$\tabularnewline
\hline 
$\text{\ensuremath{\frac{1}{2}}}(2q+1)\leq m<\text{\ensuremath{\frac{2}{2}}}(2q+1)$ & $1$\tabularnewline
\hline 
$\frac{2}{2}(2q+1)\leq m<\text{\ensuremath{\frac{3}{2}}}(2q+1)$ & $2$\tabularnewline
\hline 
$\frac{3}{2}(2q+1)\leq m<\text{\ensuremath{\frac{4}{2}}}(2q+1)$ & $3$\tabularnewline
\hline 
$\frac{4}{2}(2q+1)\leq m<\text{\ensuremath{\frac{5}{2}}}(2q+1)$ & $4$\tabularnewline
\hline 
\end{tabular}
\par\end{center}

\medskip{}

\begin{center}
\begin{tabular}{|c|c|}
\hline 
interval & $\left\lfloor \frac{m}{2q+1}\right\rfloor $\tabularnewline
\hline 
\hline 
$0\leq m<2q+1$ & $0$\tabularnewline
\hline 
$2q+1\leq m<2(2q+1)$ & $1$\tabularnewline
\hline 
$2(2q+1)\leq m<3(2q+1)$ & $2$\tabularnewline
\hline 
$3(2q+1)\leq m<4(2q+1)$ & $3$\tabularnewline
\hline 
$4(2q+1)\leq m<5(2q+1)$ & $4$\tabularnewline
\hline 
\end{tabular}
\par\end{center}

\medskip{}

\begin{center}
\begin{tabular}{|c|c|}
\hline 
interval & $\left\lfloor \frac{2(m+1)}{2q+1}\right\rfloor $\tabularnewline
\hline 
\hline 
$0\leq m<\text{\ensuremath{\frac{1}{2}}}(2q+1)-1$ & $0$\tabularnewline
\hline 
$\text{\ensuremath{\frac{1}{2}}}(2q+1)-1\leq m<\text{\ensuremath{\frac{2}{2}}}(2q+1)-1$ & $1$\tabularnewline
\hline 
$\frac{2}{2}(2q+1)-1\leq m<\text{\ensuremath{\frac{3}{2}}}(2q+1)-1$ & $2$\tabularnewline
\hline 
$\frac{3}{2}(2q+1)-1\leq m<\text{\ensuremath{\frac{4}{2}}}(2q+1)-1$ & $3$\tabularnewline
\hline 
$\frac{4}{2}(2q+1)-1\leq m<\text{\ensuremath{\frac{5}{2}}}(2q+1)-1$ & $4$\tabularnewline
\hline 
\end{tabular}
\par\end{center}

\medskip{}

\begin{center}
\begin{tabular}{|c|c|}
\hline 
interval & $\left\lfloor \frac{m+1}{2q+1}\right\rfloor $\tabularnewline
\hline 
\hline 
$0\leq m<(2q+1)-1$ & $0$\tabularnewline
\hline 
$(2q+1)-1\leq m<2(2q+1)-1$ & $1$\tabularnewline
\hline 
$2(2q+1)-1\leq m<3(2q+1)-1$ & $2$\tabularnewline
\hline 
$3(2q+1)-1\leq m<4(2q+1)-1$ & $3$\tabularnewline
\hline 
$4(2q+1)-1\leq m<5(2q+1)-1$ & $4$\tabularnewline
\hline 
\end{tabular}
\par\end{center}

\medskip{}

As we can see, there are four kinds of interval (hereinafter referred
to as cases) that need to be considered with respect to $m$. Every
analyzed interval is parametrized by the auxiliary variable $s\in\mathbb{N}\cup\{0\}$.
By choosing arbitrarily $s=0,1,2,3,\ldots$ we are able to analyze
the function $h$, and as follows $\mathcal{G}(n,m)-\mathcal{G}(n,m+1)$,
for every interesting $m$. The cases we need to consider are: 

\medskip{}

\begin{center}
\begin{tabular}{|c|c|c|c|c|c|}
\hline 
Case & interval of $m$ & $\left\lfloor \frac{2m}{2q+1}\right\rfloor $ & $\left\lfloor \frac{m}{2q+1}\right\rfloor $ & $\left\lfloor \frac{2(m+1)}{2q+1}\right\rfloor $ & $\left\lfloor \frac{m+1}{2q+1}\right\rfloor $\tabularnewline
\hline 
\hline 
$1$b & $\frac{2s}{2}(2q+1)\leq m<\frac{2s+1}{2}(2q+1)-1$ & $2s$ & $s$ & $2s$ & $s$\tabularnewline
\hline 
$2$b & $m=\frac{2s+1}{2}(2q+1)-1$ & $2s$ & $s$ & $2s+1$ & $s$\tabularnewline
\hline 
$3$b & $\frac{2s+1}{2}(2q+1)\leq m<\frac{2s+2}{2}(2q+1)-1$ & $2s+1$ & $s$ & $2s+1$ & $s$\tabularnewline
\hline 
$4$b & $m=\frac{2s+2}{2}(2q+1)-1$ & $2s+1$ & $s$ & $2s+2$ & $s+1$\tabularnewline
\hline 
\end{tabular}
\par\end{center}

\smallskip\textsc{Case 1b}\smallskip

Let $\frac{2s}{2}(2q+1)\leq m<\frac{2s+1}{2}(2q+1)-1$. 

In general $0\leq m\leq q^{2}-1$, thus $0\leq\frac{2s}{2}(2q+1)$
and $q^{2}-1<\frac{2s+1}{2}(2q+1)-1$ which implies (providing that
$s\in\mathbb{N}\cup\{0\}$) that $0\leq s$ and $s$ should not be
greater than the smallest integer that meets the inequality $s>\frac{q^{2}}{2q+1}-1$.
This implies that $s=\left\lfloor \frac{q^{2}}{2q+1}-1\right\rfloor +1$,
thus $s\leq\frac{q^{2}}{2q+1}$. On the other hand, $\frac{2s}{2}(2q+1)\leq m$
and $m\leq q^{2}-1$. This suggests that $\frac{2s}{2}(2q+1)\leq q^{2}-1$,
i.e. $s\leq\frac{q^{2}-1}{2q+1}$. Since the second constraint is
more restrictive\footnote{as $\frac{q^{2}}{2q+1}-\frac{q^{2}-1}{2q+1}=\frac{1}{2q+1}$}
then we adopt $s\leq\frac{q^{2}-1}{2q+1}$ 

Thus, let us consider $h(q,m,r_{1},r_{2},r_{3}.r_{4})$ where, following
the assumptions of case 1, $r_{1}=2s$, $r_{2}=s$, $r_{3}=2s$ and
$r_{4}=s$. It is easy to calculate that

\begin{equation}
h(q,m,2s,s,2s,s)=12q-22s-6\label{eq:lemma-d:eq:23}
\end{equation}

The highest possible $s$ is $\frac{q^{2}-1}{2q+1}$, hence it holds
that 
\begin{equation}
h(q,m,2s,s,2s,s)\geq6(2q-1)-22\left(\frac{q^{2}-1}{2q+1}\right)\label{eq:lemma-d:eq:24}
\end{equation}

which is true if and only if

\begin{equation}
h(q,m,2s,s,2s,s)\geq\frac{2\left(q^{2}+8\right)}{2q+1}\label{eq:lemma-d:eq:25}
\end{equation}

It is clear that the above expression is strictly higher than $0$
for $q\geq1$. 

\smallskip\textsc{Case 2b}\smallskip

Let $m=\frac{2s+1}{2}(2q+1)-1$  

The highest possible value of $m$ is $q^{2}-1$ thus $m=\frac{2s+1}{2}(2q+1)-1\leq q^{2}-1$,
hence, $s\leq\frac{1}{2}\left(\frac{2q^{2}}{2q+1}-1\right)$.

Let us consider $h(q,m,r_{1},r_{2},r_{3}.r_{4})$ where (see case
2) $r_{1}=2s$, $r_{2}=s$, $r_{3}=2s+1$, $r_{4}=s$ and denote:
\begin{equation}
\widehat{h}(q,m,r_{1},r_{2},r_{3}.r_{4})\stackrel{\textit{df}}{=}h(q,\frac{2s+1}{2}(2q+1)-1,2s,s,2s+1,s)\label{eq:lemma-d:eq:26}
\end{equation}

Thus, we may calculate that

\begin{equation}
\widehat{h}(q,m,r_{1},r_{2},r_{3}.r_{4})=12q-22s-6\label{eq:lemma-d:eq:27}
\end{equation}

Adopting the upper bound of $s=\frac{1}{2}\left(\frac{2q^{2}}{2q+1}-1\right)$
we obtain

\begin{equation}
\widehat{h}(q,m,r_{1},r_{2},r_{3}.r_{4})\geq12q-22\left(\frac{1}{2}\left(\frac{2q^{2}}{2q+1}-1\right)\right)-6\label{eq:lemma-d:eq:28}
\end{equation}

which is equivalent to 

\begin{equation}
\widehat{h}(q,m,r_{1},r_{2},r_{3}.r_{4})\geq\frac{2q^{2}+22q+5}{2q+1}\label{eq:lemma-d:eq:29}
\end{equation}

It is clear that the right side of the above expression is strictly
higher than $0$ for $q\geq1$.

\smallskip\textsc{Case 3b}\smallskip

Let $\frac{2s+1}{2}(2q+1)\leq m<\frac{2s+2}{2}(2q+1)-1$. The highest
possible value of $m$ is $q^{2}-1$, thus the highest possible value
of $s$ cannot be greater than the smallest positive integer for which
$q^{2}-1<\frac{2s+2}{2}(2q+1)-1$. Hence $\frac{q^{2}}{2q+1}-2<s$,
which implies that $\left\lfloor \frac{q^{2}}{2q+1}-2\right\rfloor +1=s$.
Therefore $\frac{q^{2}}{2q+1}-1\geq s$. On the other hand, $\frac{2s+1}{2}(2q+1)\leq m$
and $m\leq q^{2}-1$. This suggests that $\frac{1}{2}\left(\frac{2\left(q^{2}-1\right)}{2q+1}-1\right)\geq s$.
Since the first condition is more restrictive\footnote{as $\frac{1}{2}\left(\frac{2\left(q^{2}-1\right)}{2q+1}-1\right)-\left(\frac{q^{2}}{2q+1}-1\right)=\frac{2q-1}{4q+2}$}
then we assume that $\frac{q^{2}}{2q+1}-1\geq s$. 

Let us consider $h(q,m,r_{1},r_{2},r_{3}.r_{4})$ where (following
case 2) $r_{1}=2s+1$, $r_{2}=s$, $r_{3}=2s+1$ and $r_{4}=s$. It
is easy to calculate that

\begin{equation}
h(q,m,2s+1,s,2s+1,s)=2(6q-11s-7)\label{eq:lemma-d:eq:30}
\end{equation}

The upper bound for $s$ is $\frac{q^{2}}{2q+1}-1$, thus 

\begin{equation}
h(q,m,2s+1,s,2s+1,s)\geq2\left(6q-11\left(\frac{q^{2}}{2q+1}-1\right)-7\right)\label{eq:lemma-d:eq:31}
\end{equation}

which is equivalent to 
\begin{equation}
h(q,m,2s+1,s,2s+1,s)\geq\frac{2\left(q^{2}+14q+4\right)}{2q+1}\label{eq:lemma-d:eq:32}
\end{equation}

It is clear that the above expression is strictly higher than $0$
for $q\geq1$.

\smallskip\textsc{Case 4b}\smallskip

Let $m=\frac{2s+2}{2}(2q+1)-1$. The highest possible value of $m$
is $q^{2}-1$. Thus $m=\frac{2s+2}{2}(2q+1)-1\leq q^{2}-1$, which
is equivalent to $s\leq\frac{1}{2}\left(\frac{q^{2}}{2q+1}-1\right)$.

Let us consider $h(q,m,r_{1},r_{2},r_{3}.r_{4})$ where (see case
4) $r_{1}=2s+1$, $r_{2}=s$, $r_{3}=2s+2$, $r_{4}=s+1$ and denote:

\begin{equation}
\widehat{h}(q,m,r_{1},r_{2},r_{3}.r_{4})\stackrel{\textit{df}}{=}h(q,\frac{2s+2}{2}(2q+1)-1,2s+1,s,2s+2,s+1)\label{eq:lemma-d:eq:33}
\end{equation}
It is easy to calculate that

\begin{equation}
\widehat{h}(q,m,r_{1},r_{2},r_{3}.r_{4})=2(6q-11s-7)\label{eq:lemma-d:eq:34}
\end{equation}

As the highest possible value of $s$ is $\frac{1}{2}\left(\frac{q^{2}}{2q+1}-1\right)$
then 
\begin{equation}
\widehat{h}(q,m,r_{1},r_{2},r_{3}.r_{4})\geq2\left(6q-11\left(\frac{1}{2}\left(\frac{q^{2}}{2q+1}-1\right)\right)-7\right)\label{eq:lemma-d:eq:35}
\end{equation}

Which is equivalent to 

\begin{equation}
\widehat{h}(q,m,r_{1},r_{2},r_{3}.r_{4})\geq\frac{13q^{2}+6q-3}{2q+1}\label{eq:lemma-d:eq:36}
\end{equation}

It is easy to verify that the above expression is strictly greater
than $0$ for $q\geq1$. The last observation completes the proof
of the lemma. \QED
\end{document}